\documentclass[9pt,draftclsnofoot,compsoc,onecolumn,twoside]{IEEEtran}

\newcommand{\G}{\mathcal{G}}
\newcommand{\N}{\mathcal{N}}
\newcommand{\A}{\mathbf{A}}
\newcommand{\vv}{\mathbf{v}}

\ifCLASSOPTIONcompsoc
  \usepackage[nocompress]{cite}
\else
  \usepackage{cite}
\fi

\ifCLASSINFOpdf
\else
\fi
\usepackage{url}
\usepackage{amsmath}
\usepackage{graphicx}
\usepackage{color}
\usepackage{hyperref}

\usepackage{todonotes}

\usepackage{amsmath,amssymb}
\usepackage{cite}
\usepackage{nameref,hyperref}
\usepackage{amsmath}
\usepackage{breqn}
\usepackage{amssymb}
\usepackage{amsthm}

\usepackage{algorithmic}
\usepackage[]{algorithm2e}

\newtheorem{theorem}{Theorem}

\newtheorem{definition}{Definition}
\newcommand{\norm}[1]{\left\lVert#1\right\rVert}

\newcommand{\blue}[1]{\textcolor{black}{#1}}

\begin{document}
%

\title{A general centrality framework based on node navigability}

%
%
%
%

\author{Pasquale~De~Meo,~
        Mark~Levene,
        Fabrizio~Messina,
        and~Alessandro~Provetti
\IEEEcompsocitemizethanks{\IEEEcompsocthanksitem P. De Meo is with the Department
of Ancient and Modern Civilizations, University of Messina. Messina,
Italy, 98166.\protect\\
E-mail: \href{mailto:pdemeo@unime.it}{pdemeo@unime.it}

\IEEEcompsocthanksitem M. Levene and A. Provetti are with the Department of Computer Science and Information Systems, Birkbeck, University of London. London WC1E 7HX, UK.\protect\\
E-mail: \href{mailto:ale@dcs.bbk.ac.uk}{\{mark,ale\}@dcs.bbk.ac.uk}

\IEEEcompsocthanksitem F. Messina is with the Department of Computer Science, University of Catania, Catania, Italy.\protect\\
E-mail: \href{mailto:messina@dmi.unict.it}{messina@dmi.unict.it} 

}

} 

\IEEEtitleabstractindextext{%
\begin{abstract}
Centrality metrics are a popular tool in Network Science to identify important nodes within a graph.  
We introduce the {\em Potential Gain} as a centrality measure that unifies many walk-based centrality metrics in graphs and captures the notion of node navigability, interpreted as the property of being reachable from anywhere else (in the graph) through short walks. 
Two instances of the Potential Gain (called the {\em Geometric} and the {\em Exponential Potential Gain}) are presented and we describe scalable algorithms for computing them on large graphs.
We also give a proof of the relationship between the new measures and established centralities.
The geometric potential gain of a node can thus be characterized as the product of its Degree centrality by its Katz centrality scores.
At the same time, the exponential potential gain of a node is proved to be the product of Degree centrality by its Communicability index.
These formal results connect potential gain to both the ``popularity'' and ``similarity'' properties that are captured by the above centralities.
\end{abstract}

} 

\maketitle

\IEEEdisplaynontitleabstractindextext

%
\IEEEpeerreviewmaketitle

\section{Introduction}
\label{sec:introduction}
Centrality metrics \cite{lu2016vital} provide a ubiquitous Network Science tool for the identification of the ``important'' nodes in a graph.
They have been widely applied in a range of domains such as early detection of epidemic outbreaks \cite{chung2009distributing}, viral marketing \cite{leskovec2007dynamics}, trust assessment in virtual communities \cite{AgresteMFPP15}, preventing catastrophic outage in power grids \cite{albert2004structural} and analysing heterogeneous networks \cite{Agreste-anobii15}.

Some centrality metrics define the importance of a node $i$ in a graph $\G$ as function of the distance of $i$ to other nodes in $\G$: for instance, in  {\em Degree Centrality}, the importance of $i$ is defined as the number of the nodes which are adjacent to $i$, i.e. which are at distance one from $i$. 
Analogously, {\em Closeness Centrality} \cite{newman2010networks,boldi2014axioms} classifies as important those nodes which are few hops away from any other node in $\G.$

Another class of centrality metrics looks at walk/path structures in $\G$: for instance, the {\em Betweenness  Centrality} \cite{newman2010networks} of a node depends on the number of shortest paths crossing it and, thus, nodes with largest betweenness centrality scores are those which intercept most of the shortest paths in $\G$.
A further popular metric is the {\em Katz's Centrality Score} \cite{katz1953new}, which is understood as the weighted number of walks terminating in the node. 
In Katz's centrality metric the weighting factor is inversely related to walk length and, thus, long (resp., short) walks have a small (resp., large) weight.
For a suitable choice of the weighting factor, the Katz centrality score converges to the {\em Eigenvector Centrality} \cite{benzi2014matrix,boldi2014axioms} and to the {\em PageRank} \cite{brin1998anatomy,boldi2014axioms}.

To the best of our knowledge, however, there is no previous work in which the centrality of a node is closely related to the notion of {\em navigability,} defined as the ease at which it is possible to reach a target node $i$ regardless of the node $j$ chosen as source node.
Navigability is one of the most important features for a broad range of natural and artificial systems which have the transportation of information (e.g. a computer network) or the trade of goods (e.g. a road network) as their primary purpose.

Early studies on graph navigability were inspired by the seminal work of Travers and Milgram \cite{travers1967small} on the ``small world'' property for social networks: in a celebrated experiment, randomly-chosen Nebraska residents were asked to send a booklet to a complete stranger in Boston. Selected individuals were required to forward the booklet to any of their acquaintances whom they deemed likely to know the recipient or at least might know people who did. In some cases, the booklet actually reached the target recipient by means, on average, of 5.2 intermediate contacts, thus suggesting an intriguing feature of human societies: in large, even planetary-scale, social networks, pairs of individuals are connected through {\em shorts chains of intermediaries} and ordinary people are able to uncover these chains \cite{kleinberg2000small,GoMuWa09,LeHo08}.
Several empirical studies have revealed the small-world effect in diverse domains such as metabolic and biological networks \cite{jeong2000large}, the Web graph \cite{broder2000graph}, collaboration networks among scientists \cite{newman2001structure} as well as social networks \cite{dodds2003experimental,watts1998collective}.

So far, centrality metrics and navigability have been investigated in parallel, yet their research tracks are disconnected. 
Thus, an important (and still unanswered) direction of inquiry is the introduction of centrality metrics that are related to the {\em navigability} of a node.

Our main contribution is to tackle the questions above by extending previous work 
by Fenner et al. \cite{fenner2008modelling} who studied navigability in the context of web surfing. 
The main output of our research is a general framework in which the Potential Gain unifies several walk-based centrality metrics and captures the notion of navigability.

The potential gain of a node $i$ depends on the number of walks $w_k(j,i)$ of length $k$ that connect $i$ with any other node $j.$ 
The underlying idea is that, for a fixed $k,$ the larger $w_k(j, i),$ the higher the chance that $j$ will reach $i,$ regardless of the specific navigation strategy. 
In our model, the contribution of each walk to the potential gain shall decrease with its length $k.$ 
This intuition is formalized by the introduction of a weighting factor $\phi(k)$ which monotonically decreases with $k$ to penalize long walks.

\noindent
We have developed two variants of the potential gain of \cite{fenner2008modelling}, namely:

\begin{itemize}
	\item the {\em Geometric Potential Gain,} in which $\phi(k)$ decays as $\delta^k$, where $\delta$ is a parameter ranging between $0$ and the inverse of the spectral radius $\lambda_1$ of $\G$%
	\footnote{Recall that the spectral radius of $\G$ is defined as the largest eigenvalue of the adjacency matrix of $\G$.}, and
	
	\item the {\em Exponential Potential Gain,} in which $\phi(k)$ decays in exponential fashion.
\end{itemize} 

Both the geometric and exponential gain of $i$ can be thought as the product of one index (Degree Centrality) related to the \emph{popularity} of $i$ and another (Katz Centrality score, for the geometric potential gain, and the Communicability Index \cite{benzi2013total,estrada2005subgraph} for the exponential potential gain) which reflects the degree of \emph{similarity} of $i$ with all other nodes in the network. 
The combination of popularity and similarity has proven to closely resemble the way humans navigate large social networks \cite{csimcsek2008navigating} or attempt to locate information in large information networks such as Wikipedia \cite{west2009wikispeedia,west2012automatic,helic2013models}.

Our formalisation applies the Neumann series expansion \cite{horn2013matrix} to efficiently, yet accurately approximate both the geometric and exponential gain.
Both theoretical and experimental analysis show that our approach is appropriate for accurately computing the geometric and exponential potential gain in large real-life graphs consisting of millions of nodes and edges, even with modest hardware resources.

We validated our approach on three large datasets: \textsc{Facebook} (a graph of friendships among Facebook users), \textsc{DBLP} (a graph describing scientific collaboration among researchers in Computer Science) and \textsc{YouTube} (a graph mapping friendship relationships among YouTube users). 

\noindent
The main findings of our study can be summarized as follows:

\begin{enumerate}
\item The amount of time needed to compute the geometric or the exponential potential gain {\em does not depend} on the number of nodes/edges of a graph; instead, it depends on the spectral radius $\lambda_1$: 
the larger $\lambda_1$, the better connected the graph and, thus, the larger the number of walks needed to get a good approximation of the geometric/exponential potential gain.

\item For small values of $\delta$, the geometric potential gain is highly correlated with Degree Centrality, while for large values of $\delta$ it displays a strong and positive correlation with Eigenvector Centrality.

\item In the case of the geometric potential gain, walks of small length (i.e., up to ten) are sufficient to obtain a good approximation. 
In contrast, to compute the exponential potential gain our algorithm needed to construct longer random walks, in some cases up to ten times longer than those required for the computation of the geometric potential gain.

\item As a consequence of the above point, the geometric potential gain seems to be the most efficient solution for large graphs. 
\end{enumerate} 

This article is organized as follows: in Section \ref{sec:background} we provide basic definitions that will be used throughout the paper. 
In Section \ref{sec:related-works} we review related work.
Section \ref{sec:network-navigability} introduces the geometric and exponential potential gain and illustrates their main properties.
In Section \ref{sec:methods} we discuss how to efficiently calculate the geometric and exponential potential gain, while Section \ref{sec:experiments} details the experiments we have performed. 
Finally, in Section \ref{sec:conc} we draw our conclusions.

\section{Background}
\label{sec:background}

In this section we introduce some basic terminology for graphs that will be largely used throughout this article.

Let a graph $\mathcal{G}$ be an ordered pair $\G = \langle N, E \rangle$ where \textit{N} is a set of {\em nodes}, here also called {\em vertices}, and $E = \{\langle i, j \rangle: i,j \in V\}$ is the set of {\em edges.} 
As usual, $\G$ is {\em undirected} if edges are unordered pairs of nodes and {\em directed} otherwise.
In this article we will consider only undirected graphs. 

Also, let $n = \vert V \vert$ be the number of nodes, $m = \vert E \vert$ the number of edges of $\G$. 
For any given node \textit{i} its neighborhood $\N(i)$ is the set of nodes directly connected to it; its {\em degree} $d_i$ is the number of edges incident onto it, i.e., $d_i = \vert \N(i) \vert$.

A {\em walk} of length $k$ (with  $k$ a non-negative integer) is a sequence of nodes $\langle i_0, i_1, \ldots, i_k\rangle$ such that consecutive nodes are directly connected: $\langle i_{\ell}, i_{\ell+1} \rangle\in E$ for $\ell \in [0..k-1].$ 
Also, we use the term {\em path} for walks that do not have repeated vertices. 
A walk will be {\em closed} if it starts and ends at the same node.


We will represent graphs by their associated {\em adjacency matrix,} $\A,$ defined as $a_{ij}=1$ if $\langle i,j\rangle \in E$ and 0 otherwise. 
Sometimes we may slightly simplify notation with $a_{ij} = \A_{ij}.$ 
The adjacency matrix provides a compact formalism to describe many graph properties: for instance, the matrix $\A^2$ where $a^2_{ij} = \sum_{r = 1}^{n} a_{ir}a_{rj}$, gives the number of walks of length two going from $i$ to $j$. 
Inductively, for any positive integer $k$, the matrix $\A^k$ will give the number of closed (resp., distinct) walks of length $m$ between any two nodes $i$ and $j$ if $i = j$ (resp., if $i \neq j$) \cite{cvetkovic1997eigenspaces}.

It is a well-know fact that the adjacency matrix of any undirected graph is {\em symmetric} and, hence, all its eigenvalues $\lambda_1 \geq \lambda_2 \geq \ldots \geq \lambda_n$ are real. 
The largest eigenvalue $\lambda_1$ of $\A$ is also called its {\em principal eigenvalue} or {\em spectral radius} of $\G.$
Moreover, the corresponding eigenvectors $\vv_1, \ldots, \vv_n$ will form an orthonormal basis in $\mathbb{R}^n$ \cite{strang1993introduction}.
Eigenpairs $\langle \lambda_i, \vv_i \rangle$ are formed by the eigenvalue $\lambda_i$ and the corresponding eigenvector $\mathbf{v}_i$.

\section{Related Work}
\label{sec:related-works}

The task of searching for and navigating in large networks has been extensively studied in the past for a broad range of domains such as routing in small world networks \cite{travers1967small}, locating pages in the World Wide Web \cite{menczer2002growing,broder2000graph}, finding the most knowledgeable individual in an enterprise/academic social network, discovering resources in a P2P network and building recommender systems \cite{lamprecht2016method}. 
Inspired by the classification scheme introduced by Helic {\em et al.}, \cite{helic2013models}, we divide search and navigation into two main classes: {\em (a) Endogenous Search}. 
In this class, there are {\em multiple} agents embedded in the network and the navigation task is depicted as a decentralized decision process in which agents collaborate to discover a path in the network.
Agents are assumed to have only a local knowledge of the network topology and, in addition, querying a neighboring node (e.g., to route a message) may have a non-negligible cost. {\em (b) Exogenous Search}. 
This class occurs whenever a user aims at navigating the Web \cite{fenner2008modelling,levene2004navigating} or an information network such as Wikipedia \cite{west2012automatic,lamprecht2017structure}. 
In exogenous search, there is {\em only one} agent involved in navigation task and it does not belong to the network. 
As in endogenous search, the agent (either human and artificial) only posses local knowledge about the network topology.
Unlike the endogenous search, however, the cost for visiting a node is generally low.

In the following two sections we review methods in endogenous and exogenous search (see Sections \ref{sub:endogenous-search} and \ref{sub:exogenous-search} below).

\subsection{Endogenous search approaches}
\label{sub:endogenous-search}


Many mathematical models have been proposed to explain why networks are, in an informal sense, \textit{navigable.}
Some of the best-known models are described in \cite{kleinberg2000navigation,kleinberg2000small,kleinberg2002small,watts2002identity}.
The original Watts-Strogatz (WS) model \cite{watts1998collective} generated random graphs in which pairs of nodes belonging to distant parts of the graph may be connected through random edges, called {\em long-range weak ties.} 
The WS model was thus effective in forcing the graph to be ``small,'' i.e., to assure the presence of paths consisting of few edges between any pair of nodes. 
Nevertheless, the WS model alone is unable to explain why people are capable of discovering such paths.

One class of approaches to search large networks relies on the notion of {\em popularity.} 
Adamic et al. \cite{adamic2001search} is perhaps the best-known approach in that category: here the search task is modelled as a random walk on the graph $\G$. 
In a specific step, if the walker occupies a node $l$ which is not
the target one (and none of the neighbours of $l$ is the target one), then the walker chooses the unvisited neighbour with largest degree. 

Another class of approaches -- called {\em similarity-based approaches} -- exploit homophily to speed up search tasks \cite{kleinberg2000small,watts2002identity}.
Kleinberg \cite{kleinberg2000navigation} described a generalization of the WS model to explain why decentralized search is effective in real networks. 
In Kleinberg's model, nodes of a social network are arranged to form a bi-dimensional grid (called {\em a lattice);} each node is connected to its neighbours in the lattice and, in particular, the distance $d(u, v)$ between two nodes equals the number of grid steps separating them. 
As a result, each node $v$ is connected to its four local contacts (i.e., nodes at distance one from $v$).
In addition, a random edge---called a {\em long range edge}---connecting $v$ with a node $w$ is generated with probability proportional to $d(v,w)^{-q}$, $q$ being the so-called clustering exponent of the model.
Kleinberg proved that if $q = 2$, then the performance of decentralized search is optimal, i.e. there exists an algorithm that, on average, is able to deliver a message from an arbitrary source node to an arbitrary target in $O(\log^2 n)$ time. 

Another contribution is due to Watts {\em et al.} \cite{watts2002identity} who proposed a model to explain network navigability in which nodes aggregate into {\em groups} on the basis of some shared attributes such as job or geographic location. 
For each attribute, a population can be split into a hierarchical set of layers $\mathcal{H}$ in which the top layer describes the entire population, while layers at increasing depth define a cognitive division of the population into more specific groups. 
Individuals can manage two kinds of information to decide to whom a message should be forwarded to:  first, {\em social distance}, which accounts for the similarity of two nodes and, second, {\em network distance}, i.e., the number of network paths that can be detected by looking at the neighbors. 

Social distance between two individuals $i$ and $j$ can be estimated by considering the groups $i$ and $j$ belong to and how distant these groups are in the layers of $\mathcal{H}$. 
As such, social distance is a kind of global metrics but, unfortunately, it is not a true distance in the sense that individuals belonging to close groups may be separated by long paths in the social network graph. 

Network distance, on the other hand, is a true distance but a node only has access to a local portion of the network and, thus, it can correctly calculate network distances only for nodes which are separated from it by a few hops.
By means of simulations \cite{watts2002identity} demonstrated that social distance is effective in approximating network distance and, thus, to successfully in directing messages across the network.

More recently Csimcsek {\em et al.} \cite{csimcsek2008navigating} suggest to combine popularity- and similarity-based methods and, to this end, they describe a graph search algorithm that, at each step, forwards a message from the current node $i$ to one of its neighbours, say $j$, such that the product $d_i \times q_{ij}$ is maximum; here $q_{ij}$ quantifies homophily between $i$ and $j$.

\blue{Unlike the approaches above, which assume that some attributes are available at each node, our approach only makes use of connectivity patterns to calculate the centrality of a node.}

\subsection{Exogenous search approaches}
\label{sub:exogenous-search}

Exogenous search is mostly related to search in information networks such as a collection of Web pages or Wikipedia.

One of the most common search strategies on the World Wide Web (WWW) is {\em surfing}, in which a user moves from a Web page to another one by following hyperlinks. 
Huberman {\em et al.} \cite{huberman1998strong} introduced a probabilistic model to describe surfing. 
In this model, the sequence of Web pages a user visits is regarded as the realisation of a random process and each Web page is associated with a value to the user. 
A user will stop surfing if the estimated cost of accessing a new Web page is bigger than the expected value of the information the user may get from accessing it. 

More recently, West and Leskovec \cite{west2012human} analysed how people navigate an information network such as Wikipedia in order to reach a specific target. 
To this end, they used an online computation game, called {\em Wikispeedia} \cite{west2009wikispeedia}, in which Wikipedia information seekers are given two random articles and they are required to navigate from one to the other by clicking as few hyperlinks as possible. 
In a subsequent paper \cite{west2012automatic}, they compared the accuracy of several decentralized search algorithms and benchmarked them  against  the  human  navigation  paths.  
Such a study highlighted two main phases of human navigation in information networks: 
{\em (i) Zoom-Out}: here, users strives to reach the network {\em core} (or a hub in the network core); such a core consists of a Wikipedia page with many links to other pages in Wikipedia. 
In this step, humans would prefer pages with many outgoing links ({\em high degree pages}). 
{\em (ii)  Zoom-in}, in  which  users leave the core to get closer to a topic. Specifically, if we think of segmenting Wikipedia pages into clusters on the basis of their topics, such a phase would consist of entering into a cluster. In the zoom-in phase, users prefer to look for similar nodes in order to orient their search.

A nice approach to combine decentralized search was described by Helic {\em et al.} \cite{helic2013models} who applied decentralized search algorithms such as those described in \cite{kleinberg2000navigation} to model human navigation in information networks. 
They considered an online navigation game (called {\em WikiGame}); in this game, a user starts from a random Wikipedia page and navigates to a target page.
More than 250,000 click paths were collected  and studied to determine the factors influencing players' decisions. 
The main finding in \cite{helic2013models} is that two mechanisms regulate the way humans seek for information in large networks: i) {\em exploitation,} i.e., humans follow specific hyperlinks whenever they are confident enough that those will get them closer to the target they want, and ii) {\em exploration}, i.e., users navigate at random an information network, when their knowledge about how current links relates to a target Web page is insufficient. 
The quantitative analysis showed that exploration steps account, on average, for $15-20\%$ of collected links, while exploitation accounted for the remaining $80-85\%$ of collected links.

\section{A model of network navigability}
\label{sec:network-navigability}

In this section we introduce our notion of network navigability, in Subsection \ref{sub:navigability-potential} along with the Potential Gain as a new centrality measure and a new general framework that unifies many walk-based centrality indices and captures our notion of navigability. 
In Section \ref{sub:geometric-exponential-potantial-gain} we describe two versions of the potential gain, called the {\em geometric} and the {\em exponential} potential gains.
In Subsection \ref{sub:relation-to-centrality} we compare the two versions of potential gain with other, well-known, centrality metrics from the literature.
In Subsection \ref{sub:relation-geometric-exponential} we investigate the relationship between the geometric and the exponential potential gain.
Finally, \ref{sub:fast-calculation-geom-expon} outlines our approach to calculating the geometric and exponential potential gain.

\subsection{Network navigability and the potential gain}
\label{sub:navigability-potential}

Starting from the Travers-Milgram's experiment \cite{travers1967small}, several studies have sought to characterize under which conditions a graph is deemed {\em navigable}. 
According to Kleinberg \cite{kleinberg2000navigation}, a graph $\G$ is {\em navigable} if i) its diameter is bounded by a polynomial in $\log n$ (here $n$ is the number of nodes in $\G$) and ii) it has a strongly-connected component which contains almost all of the nodes of $\G$. 
Kleimberg's navigability is a property of the graph, whereas our purpose is to define the notion of navigability at the {\em node level.}

For navigability at the node level we leverage the main findings of Lamprecht  et al. \cite{lamprecht2015improving} on the navigability of a recommender network built on top of the Internet Movie Database (IMDB). 
Their work introduced some topological indices to quantify how difficult it is to navigate the recommender network.
Specifically, they suggested to compute the {\em eccentricity} $ecc(j)$ of each node $j$, defined as the length of the longest shortest-path converging to $j$ from any other node belonging to the same connected component of $j$. 
Thus, nodes with small eccentricity (also termed {\em efficient reachability}  can be easily reached from any other node of the graph.

Eccentricity may seem a good starting point for the formalisation of network navigability but ``suffers'' from some known issues. 
First, the eccentricity $ecc(j)$ is dominated by the distance, from $j$ to the farthest node: as a consequence, $j$ could have a large eccentricity even if it is close to almost all of nodes in $\G$.
Second, any time the length of the shortest path connecting node $i$ to $j$ is above a threshold $\theta$, the navigation from $i$ to $j$ is considered ``hard,'' regardless of the topology.

The starting point for our work is the framework proposed by Fenner et al. \cite{fenner2008modelling} who studied the problem of identifying ``good pages'' from which to start exploring the Web. 
They classify a page $p$ as a good starting point if it satisfies the following criteria: {\em (1) it is relevant}, i.e. the content of $p$ closely matches user’s information goals, {\em (2) page $p$ is central}, i.e., the distance of $p$ to other Web pages in the Web graph is as low as possible and {\em (3) page $p$  is connected}, in the sense that $p$ is able to reach a maximum number of other pages via its outlinks.\\
A key difference between Fenner et al. and our work is that they defined the navigability score of a page/a node as its ability to act as the {\em source node} for reaching all the other nodes. 
In our setting, instead, we consider the node as the {\em target} of search, as described next.

Let us fix a source node $j$ and a target node $i$ and provide an estimate $\tau(j, i)$ of how ``easy'' it will be for $i$ to be reached if we choose $j$ as source node. 
Intuitively, the larger the number of walks from $j$ to $i$, the easier it would be to reach $i$ when starting from $j$.
In addition, assume that the task of exploring a graph is costly and such a cost increases as the length of the walks/paths we use for exploration purposes increases. 
Therefore, shorter walks should be preferred to longer ones.
By combining the requirements above, we obtain:

\begin{equation}\label{eqn:tau-ceoff}
\tau(j, i) = \sum_{k=1}^{+\infty}\phi(k)\cdot w_k(j,i)
\end{equation}

\noindent 
here $w_k(j,i)$ is the number of walks of length $k$ going from $j$ to $i$ and the non-increasing function $\phi(k)$ acts as penalty for longer walks. 
If we sum over all possible source nodes $j$, we obtain a global centrality index $pg(i)$ for $i$:

\begin{equation}\label{eqn:potential-as-sum}
pg(i) = \sum_{j \in N} \tau(j,i).
\end{equation}

\noindent
In analogy to Fenner et al., let $pg(i)$ be the {\em potential gain} of $i$.

In conclusion, the main differences between Potential Gain and Eccentricity are : {\em i)} the computation of PG is grounded on walks while the computation of the eccentricity is based on paths, 
{\em ii)} the potential gain consider all walks converging to a target node $j$ while the eccentricity consider only shortest paths reaching $j$ and, 
{\em iii)} in the potential gain, the contribution of a walk of length $k$ is penalized by a factor $\phi(k)$, while in the calculation of the eccentricity we take the length of the shortest path as is (i.e., with no penalization).

\subsection{The geometric and exponential potential gain}
\label{sub:geometric-exponential-potantial-gain}

Given the above specifications, we first define the potential gain in matrix notation.
For the base case, consider walks of length k=1, i.e., direct connections. 
Only the neighbours of a node $i$ will contribute to the potential gain of $i,$ which leads to the trivial conclusion that, at $k= 1$, nodes with the largest degree are also those ones with the largest potential gain. 

We define the vector $\mathbf{p}$ such that $\mathbf{p}_i = pg(i)$ for every node $i$:

\begin{equation}
\label{eqn:potential-matrix}
\mathbf{p} = \phi(1)\cdot \mathbf{A} \times \mathbf{1}.
\end{equation}

\noindent
If we include walks of length two, then we have to consider the squared adjacency matrix $\mathbf{A}^2$. 
So, we add a contribution $\phi(2)\cdot \mathbf{A}^2 \times \mathbf{1}$ to the potential gain.

By induction, nodes capable of reaching from $i$ through walks of length up to $k$ provide a contribution to the potential gain equal to $\phi(k)\cdot \mathbf{A}^k \times \mathbf{1}$. 
By summing over all possible values of $k$ we get to the following expression for $\mathbf{p}$:

\begin{dmath}
\label{eqn:potential-gain-matrix}
\mathbf{p} = \phi(1)\mathbf{A} \times \mathbf{1} + \phi(2)\mathbf{A}^2 \times \mathbf{1} + \ldots + \phi(k)\mathbf{A}^k \times \mathbf{1} + \ldots= {\sum_{k=1}^{+\infty} \left(\phi(k)\mathbf{A}^k\times \mathbf{1}\right) = \left(\sum_{k=1}^{+\infty}\phi(k)\mathbf{A}^k\right) \times \mathbf{1}}
\end{dmath}

To attenuate the effect of the walks' length, we will consider two weighting functions, namely:

\begin{enumerate}
\item {\em Geometric:} $\phi(k) = \delta^{k-1}$ with $\delta \in (0,1)$. 
So we define the {\em geometric potential gain,} $\mathbf{g}$:

\begin{equation}
\label{eqn:geometric-potential-gain}
\mathbf{g} = \left(\mathbf{A} + \delta \mathbf{A}^2 + \ldots + \delta^{k-1} \mathbf{A}^k + \ldots \right) \times\mathbf{1}
\end{equation} 

\item {\em Exponential:} $\phi(k) = \frac{1}{(k-1)!}$. 
So we define the {\em exponential potential gain,} $\mathbf{e}$:

\begin{equation}
\label{eqn:exponential-potential-gain}
\mathbf{e} = \left(\mathbf{A} + \mathbf{A}^2 + \ldots + \frac{1}{\left(k -1 \right) !} \mathbf{A}^k + \ldots\right) \times \mathbf{1}
\end{equation} 

\end{enumerate}

\subsection{Relation to centrality measures}
\label{sub:relation-to-centrality}

The geometric and the exponential potential gain introduced above yield a {\em ranking} of network nodes and, therefore, it is instructive to compare them with popular centrality metrics.
Recall that we defined the {\em spectral radius} $\lambda_1$ of $\mathbf{A}$ as the largest eigenvalue of $\mathbf{A}$.

As for the geometric potential gain, if we let $\delta < \lambda_1^{-1}$, the following expansion holds:

\begin{dmath}
\label{eqn:geometric-potential-series}
\mathbf{g} = \left(\mathbf{A} + \delta \mathbf{A}^2 + \ldots + \delta^{k-1} \mathbf{A}^k + \ldots \right) \times\mathbf{1} = 
\mathbf{A} \times \left(\mathbf{I} + \delta \mathbf{A} + \ldots + \delta^{k-1} \mathbf{A}^{k-1} + \ldots \right) \times\mathbf{1}\\
= \mathbf{A} \times \left(\mathbf{I} - \delta \mathbf{A}\right)^{-1} \times\mathbf{1}
\end{dmath}

\noindent 
in which we make use of the {\em Neuman series} \cite{horn2013matrix}:
  
\begin{equation}
\label{eqn:neuman-series}
\left(\mathbf{I} +  \ldots + \delta^{k-1} \mathbf{A}^{k-1} + \ldots \right) = \left(\mathbf{I} - \delta \mathbf{A}\right)^{-1}.
\end{equation}

At this point, the term $\left(\mathbf{I} - \delta \mathbf{A}\right)^{-1} \times \mathbf{1}$ is exactly the {\em Katz centrality score} \cite{katz1953new,leicht2006vertex}, a popular centrality metric that defines the importance of a node as a function of its similarity with other nodes in $\G.$
Hence, we can say that the geometric potential gain combines two kind of contributions: {\em popularity,} as captured by node degree, and {\em similarity} as captured by Katz's similarity score. 


It is also instructive to consider what happens for extreme values of $\delta$: if $\delta \to 0$, then the geometric potential gain tends to $\mathbf{A} \times \mathbf{1}$, i.e., it coincides with degree.
In contrast, if $\delta \to \frac{1}{\lambda_1}$, then the Katz centrality score converges to {\em eigenvector centrality} \cite{benzi2014matrix}, another popular metric adopted in Network Science. 
\blue{Boldi and Vigna \cite{boldi2014axioms} show that the Katz Centrality score is also strictly related to the PageRank. 
More specifically, the PageRank vector $\mathbf{p}$ coincides with the Katz Centrality score provided that the adjacency matrix $\mathbf{A}$ is replaced by its row-normalized version $\overline{\mathbf{A}}$:
\begin{equation}
\mathbf{p} = \left(1 - \alpha\right) \sum_{k=0}^{+\infty} \alpha^i \overline{\mathbf{A}}^i \times \mathbf{1}
\end{equation} 
Here, the parameter $\alpha$ is the so-called PageRank {\em damping factor}.}
Let us now concentrate on the exponential potential gain. We rewrite Equation \ref{eqn:exponential-potential-gain} as follows:

\begin{dmath}
\label{eqn:exponential-potential-series}
\mathbf{e} = \left(\mathbf{A} + \mathbf{A}^2 + \ldots + \frac{1}{\left(k -1 \right) !} \mathbf{A}^k + \ldots\right) \times \mathbf{1}
= \mathbf{A} \times \left(\mathbf{I} + \mathbf{A} + \ldots + \frac{1}{k !} \mathbf{A}^k + \ldots\right) \times \mathbf{1} 
= \mathbf{A} \times \exp(\mathbf{A}) \times \mathbf{1}
\end{dmath}

\noindent
where $\exp(\mathbf{A}) = \sum_{k=1}^{+\infty} \frac{1}{k!} \mathbf{A}^k$ is the exponential of $\mathbf{A}$ \cite{higham2008functions}. 

The exponential of a matrix has been used to introduce other centrality scores such as {\em communicability} or {\em subgraph centrality} \cite{estrada2012physics,benzi2014matrix}. 

Specifically, $\exp\left(\mathbf{A}\right)_{ij}$ measures how easy is to send a unit of flow from a node $i$ to a node $j$ and vice versa.
Such a parameter is known as {\em communicability} and it can be regarded as a measure of similarity between a pair of nodes. 
Communicability has been successfully used to discover communities in networks \cite{estrada2012physics}. 
The product $\exp(\mathbf{A}) \times \mathbf{1}$ yields a centrality metric which defines the importance of a node as function of its ability to communicate with all other nodes in the network.
In turn, the diagonal entry $\exp\left(\mathbf{A}\right)_{ii}$ of the matrix exponential defines a further centrality metric called {\em subgraph centrality} \cite{estrada2005subgraph}. 
As a result of the rewriting above, we clearly see the exponential potential gain as dependent on two factors: popularity of $i$ (i.e., its degree) and similarity of $i$ with all other nodes in the network.

The computation of the geometric (resp., exponential) potential gain for {\em all nodes} in $\G$ needs the specification of the full adjacency matrix $\mathbf{A}$; in this sense, the geometric and the exponential potential gain should be considered as {\em global centrality metrics}, on par with the Katz centrality score and subgraph centrality.

\subsection{The relation between the geometric and the exponential potential gain}
\label{sub:relation-geometric-exponential}

In this section we present some guidelines on how to choose the $\delta$ factor discussed in the previous section. 

A straightforward choice would be to set $\delta = (2\lambda_1)^{-1}$ as in \cite{katz1953new} or, in analogy with the Google PageRank damping factor, $\delta =0.85\lambda_1^{-1}$ \cite{benzi2013total}. On the other hand, Foster et al. \cite{foster2001faster} suggested the following:

$$
\delta = \frac{1}{\norm{\mathbf{A}}_{\infty} + 1}
$$

\noindent 
where $\norm{\mathbf{A}}_{\infty} = \max_{1 \leq i \leq n} \sum_{j= 1}^{n} \vert \mathbf{A}_{ij}\vert$.

It is instructive to investigate the existence of a  {\em crossover point} $\delta^{c}$, i.e. to discover a value of $\delta$ at which the geometric and the exponential gain of a node $i$ coincide. To this end, we provide the following result.

\begin{theorem}
Let $\G$ be a graph with adjacency matrix $\mathbf{A}$ and eigenvalues $\lambda_1 \geq \lambda_2 \geq \dots \geq \lambda_n$. 
For each node $i$ and for $\delta \in (0, \lambda_1^{-1})$, the geometric and the exponential gains of $i$ coincide if and only if one of the following holds:

\begin{enumerate}
	\item $\lambda_i = 0$, or

	\item $\delta = \delta^{c} = \frac{e^{\lambda_i} - 1}{\lambda_ie^{\lambda_i}}$, provided that $\delta^{c} < \lambda_1^{-1}$.
\end{enumerate}
\end{theorem}

\begin{proof}
Recall that for sufficiently large values of $k$, we can approximate the geometric and the exponential gain as follows:

\[
\mathbf{g} = \mathbf{A} \times \left(\mathbf{I} - \delta \mathbf{A}\right)^{-1} \times\mathbf{1} \quad \mbox{and} \quad \mathbf{e} = \mathbf{A} \times \exp(\mathbf{A}) \times \mathbf{1}
\]

Recall that $\mathbf{A}$ is a square and symmetric matrix.
Thus, it admits the following eigendecomposition, 

\[
\mathbf{A} = \mathbf{D}^{-1} \times \mathbf{\Lambda} \times \mathbf{D}
\]

where $\mathbf{\Lambda}$ is a diagonal matrix storing the eigenvalues $\lambda_1, \lambda_2, \ldots, \lambda_n$ of $\mathbf{A}$ and $\mathbf{D}$ is an orthonormal matrix, whose columns coincide with the eigenvectors $\mathbf{u}_1, \mathbf{u}_2, \ldots \mathbf{u}_n$ of $\mathbf{A}$.

Now recall \cite{higham2008functions} that, for any function $f$, the matrix $f(\mathbf{A})$ is still diagonalisable and, for any eigenvalue $\lambda_i$ of $\mathbf{A}$ we have that $f(\lambda_i)$ is an eigenvalue of $f(\mathbf{A})$. 
In addition, the matrices $\mathbf{A}$ and $f(\mathbf{A})$ share the same eigenvectors so we have $f(\mathbf{A}) = \mathbf{D}^{-1} \times f(\mathbf{\Lambda}) \times \mathbf{D}$.

Let us consider now the application of the two functions $f_1(x) = \frac{x}{1 - \delta x}$ and $f_2(x) = xe^x$ to matrix $\mathbf{A}$. The eigenvalues of the matrix $f_1(\mathbf{A}) = \mathbf{A} \times \left(\mathbf{I} - \delta \mathbf{A}\right)^{-1}$ are

\begin{equation}
\label{eqn:lambda_gpg}
\frac{\lambda_1}{1 - \delta \lambda_1}, \frac{\lambda_2}{1 - \delta \lambda_2}, \ldots, \frac{\lambda_n}{1 - \delta \lambda_n}
\end{equation}

whereas the eigenvalues of the matrix  $f_2(\mathbf{A}) = \mathbf{A} \times \exp(\mathbf{A})$ are 

\begin{equation}
\label{eqn:lambda_epg}
\lambda_1e^{\lambda_1}, \lambda_2e^{\lambda_2}, \ldots, \lambda_ne^{\lambda_n}.
\end{equation}

Let us introduce $\mathbf{\Lambda}_g$ and $\mathbf{\Lambda}_e$, the diagonal matrices storing the eigenvalues of the matrices $\mathbf{A} \times \exp(\mathbf{A})$ and $\mathbf{A} \times \left(\mathbf{I} - \delta \mathbf{A}\right)^{-1}$, respectively. 
Let us now compute the difference between the geometric and potential gain:

\begin{dmath}
\label{eqn:difference-geometric-potential}
\mathbf{g} - \mathbf{e} = \mathbf{A} \times \left(\mathbf{I} - \delta \mathbf{A}\right)^{-1} \times\mathbf{1} - \mathbf{A} \times \exp(\mathbf{A}) \times \mathbf{1} = \mathbf{D}^{-1} \times \mathbf{\Lambda_g} \times \mathbf{D} \times \mathbf{1} -\mathbf{D}^{-1} \times \mathbf{\Lambda_e} \times \mathbf{D} \times \mathbf{1} =
\left(\mathbf{D}^{-1} \times \left(\mathbf{\Lambda_g} - \mathbf{\Lambda_e}\right) \times \mathbf{D}\right) \times \mathbf{1}
\end{dmath}

We focus on the $i$-th component of vector $\mathbf{g} - \mathbf{e}$ and observe that its value $\Delta_i$ is given as:

\[
\Delta_i = \left(\lambda_i e^{\lambda_i} - \frac{\lambda_i}{1 - \delta \lambda_i}\right) \mathbf{u}_i^T\mathbf{u}_i =  \left(\lambda_i e^{\lambda_i} - \frac{\lambda_i}{1 - \delta \lambda_i}\right)
\]

Here we used the fact that eigenvectors of $\mathbf{A}$ form an orthonormal basis.
If we assume that $\lambda_ i \neq 0$, then $\Delta_i = 0$ if and only if:

\begin{equation}
\label{eqn:delta-crossover}
\delta = \frac{e^{\lambda_i} - 1}{\lambda_ie^{\lambda_i}}
\end{equation}

\noindent
which completes the proof.
\end{proof}

\subsection{Calculation of geometric and exponential potential Gains}
\label{sub:fast-calculation-geom-expon}

In this section we present our algorithm for the computation of the geometric potential and exponential potential gain. 

\blue{Our algorithm can be implemented in few lines of code in any high-level programming language, since it applies the expansion series provided in Equations \ref{eqn:geometric-potential-series} and \ref{eqn:exponential-potential-series}.
Our approach provides insight on how the walk length $k$ affects the calculation of the geometric (resp., exponential) potential gain: indeed, if we stop the expansion of Equation \ref{eqn:geometric-potential-series} (resp. Eq. \ref{eqn:exponential-potential-series}) after the first $k$ terms, then, we would only include the walks up to length $k$ in the calculation of the geometric (resp., exponential) potential gain.}

Let us consider the computational complexity of our solution.
\blue{We begin with the} geometric potential gain and assume that we stop expanding the Neumann series after generating walks of length $k^{\star}$.
In such a case, it is easy to see that cost will be in $O(k^{\star}\vert E\vert)$. 
\blue{In fact, for any $j$ such that $1 < j < k^{\star}$, let us set 
$\mathbf{y}_j = \delta^{j-1}\mathbf{A}^{j} \times \mathbf{1}$ and 
suppose that we have stored the sequence $\mathcal{Y} = \{\mathbf{y}_1, \mathbf{y}_2, \ldots, \mathbf{y}_{j-1}\}$, with $\mathbf{y}_1 = \mathbf{1}$, $\mathbf{y}_2 = \mathbf{A} \times \mathbf{1}$.}   

\noindent
Hence, the following recurrence holds:
 
\begin{dmath}
\label{eqn:updates-geometric-potential-gain}
\mathbf{y}_j
= 
\delta^{j-1} \mathbf{A}^{j}\times \mathbf{1} 
= 
\left(\delta \mathbf{A}\right) \times \left(\delta^{j-2} \mathbf{A}^{j-1}\times \mathbf{1}\right)
= 
\left(\delta \mathbf{A}\right) \times \mathbf{y}_{j-1}
\end{dmath}

The last equality states that any term $\mathbf{y}_j$ can be calculated as the product of a sparse matrix ($\delta \mathbf{A}$) by a vector ($\mathbf{y}_{j-1}$), already computed in the previous iteration. 
Such an operation takes $O(\vert E \vert)$ steps which, in the case of sparse networks, is $O(n)$.

Similarly, given that $\mathbf{g}$ can be expressed as $\mathbf{g} \simeq \sum_{j=0}^{k^{\star}} \mathbf{y}_j$, we conclude that the cost required to compute the geometric potential gain amounts is $O(k^{\star}n)$.
As for space complexity, the cost for computing $\mathbf{g}$ is $O(\vert E\vert)$.

The computation of the geometric potential gain requires to fix $\delta$ beforehand, which, in turn, requires to fix an approximation of the spectral radius $\lambda_1$. 
The literature on the estimation of $\lambda_1$ provides some bounds on it \cite{das2004some,stevanovic2015spectral} but, available upper bounds are often not tight and, thus, uninformative; therefore, an alternate way to approximate $\lambda_1$ is to rely on algorithms such as the {\em Power Iteration Method} \cite{heath2002scientific}. 
On the other hand, if we target very large graphs, {\em sampling techniques} seem the best option \cite{han2017closed}.

Analogous results for both time and space complexity hold for the computation of the exponential potential gain as we show next. 
Define a sequence $\mathcal{Z} = \{\mathbf{z}_i\}$ recursively as follows:

\begin{align*}
\mathbf{z}_1 = \mathbf{1} \\
\mathbf{z}_2 = \mathbf{A} \times \mathbf{1}\\
\dots\\
\mathbf{z}_i = \frac{1}{i -2}\mathbf{A}\times\mathbf{z}_{i-1}
\end{align*}

Therefore, any term $\mathbf{z}_j$ can be calculated as the product of a sparse matrix ($\mathbf{A}$) by a vector ($\mathbf{z}_{j-1}$), which has been already computed in the previous iteration. 
Such an operation takes $O(\vert E \vert)$.

Given that $\mathbf{e}$ can be expressed as $\mathbf{e} \simeq \sum_{j=0}^{k^{\star}} \mathbf{z}_j$, we can conclude that the (worst-case) time complexity for the calculation of the exponential potential gain is $O(k^{\star}\vert E \vert)$; similarly the space complexity is $O(\vert E \vert)$, hence for sparse graphs both time and space complexity reduce to $O(n).$

\section{Methods for computing the Geometric and Exponential Potential Gain}
\label{sec:methods}

\blue{In this section we prove that our algorithm to calculate the geometric and the exponential potential gain is convergent, and we provide an upper bound on the rate of convergence.}

Let $\mathbf{g}$ be the true value of the geometric potential gain and let $\mathbf{g}_k$ be the approximate value of geometric potential gain we would obtain by considering walks up to length $k$.
We wish to estimate:

\begin{equation}
\label{eqn:varepsilon-geometric}
\varepsilon_g(k) = \frac{\norm{\mathbf{g} - \mathbf{g}_k}}{\norm{\mathbf{g}}}
\end{equation}

\noindent
Analogously, the approximation error associated with the calculation of the exponential potential gain is given by

\begin{equation}
\label{eqn:varepsilon-exponential}
\varepsilon_e(k) = \frac{\norm{\mathbf{e} - \mathbf{e}_k}}{\norm{\mathbf{e}}}
\end{equation}

The evaluation of $\varepsilon_g(k)$ and $\varepsilon_e(k)$ requires us to evaluate the {\em norm} of some matrices; here we will rely on the $L_2$ {\em matrix norm} (also known as {\em the spectral norm}), which, in case of symmetric matrices coincides exactly with $\lambda_1$ \cite{heath2002scientific}. 

Since all matrix norms defined over a space of finite-dimension matrices are equivalent, our results generalize to other matrix norms; the only requirement is that the {\em sub-multiplicative property} holds, i.e., $\norm{\mathbf{X}\times \mathbf{Y}} \leq \norm{\mathbf{X}}\norm{\mathbf{Y}}$ for any pair of matrices $\mathbf{X}$ and $\mathbf{Y}$.

\subsection{Rate of convergence for the geometric potential gain}
\label{sub:rate-convergence-geometric}

Regarding the assessment of $\varepsilon_g(k)$, we note that the source of error in approximating the geometric potential gain depends on the early stopping of the Neumann series, i.e., on the approximation:

\begin{equation}
\label{eqn:early-stopping-neuman}
\left(\mathbf{I} -\delta\mathbf{A}\right)^{-1} \simeq \mathbf{I} + \delta\mathbf{A} + \delta^2\mathbf{A}^2 + \ldots
\delta^k\mathbf{A}^k  
\end{equation}
 
\noindent
Now, if we set $\mathbf{S}_k = \sum_{i=0}^{k} \delta^i\mathbf{A}^i$, the error $\varepsilon_g(k)$ depends on:

\begin{dmath}
\left( \mathbf{I} -\delta\mathbf{A}\right)^{-1} - \mathbf{S}_k = \left(\delta\mathbf{A}\right)^k +
\left(\delta\mathbf{A}\right)^{k+1} + \ldots = \left(\delta\mathbf{A}\right)^k \times\left(\mathbf{I} + \delta\mathbf{A} + \ldots \right) =
 \left(\delta\mathbf{A}\right)^k\left( \mathbf{I} -\delta\mathbf{A}\right)^{-1} 
\end{dmath}

\noindent
As $k \to \infty$ we obtain:

\begin{dmath}
\norm{\left( \mathbf{I} -\delta\mathbf{A}\right)^{-1} - \mathbf{S}_k}^{\frac{1}{k}} = \norm{\left(\delta\mathbf{A}\right)^k\left( \mathbf{I} -\delta\mathbf{A}\right)^{-1}}^{\frac{1}{k}} \leq
\norm{\delta^k\mathbf{A}^k}^{\frac{1}{k}}\norm{\left( \mathbf{I} -\delta\mathbf{A}\right)^{-1}}^{\frac{1}{k}}
\end{dmath}

Moreover, as $k \to +\infty$, $\norm{\delta^k\mathbf{A}^k}^{\frac{1}{k}}$ converges to $\lambda_1$ \cite{cvetkovic1997eigenspaces}.
This result, however, is rather weak as it does not give us a realistic estimation of the number of iterations that are required to assure that $\varepsilon_g(k) \le \varepsilon$, for any $\varepsilon > 0$.

A more refined estimation of the rate of convergence of $\varepsilon_g(k)$ that applies to the general case of square complex matrices is due to Young \cite{young1981rate}, who provides a bound of the form $O(\lambda_1^{k-n} k^n)$; it depends on $\lambda_1$, on the number \textit{k} of iterations and, finally, on the size $n$ of $\mathbf{A}$.

Since we are dealing with symmetric matrices, we can derive simpler bounds that are independent of the matrix size, as proved below.

\begin{theorem}
\label{the:convergence-geometric}
Let $\G$ be a graph with adjacency matrix $\A$ and let $\lambda_1$ be its spectral radius; 
also let $\delta \in \left(0, \lambda_1^{-1}\right)$. 
Then $\varepsilon_g(k) \to 0$ with convergence rate $\left(\delta\lambda_1\right)^k$.
\end{theorem}

\begin{proof}
Recall that matrix $\mathbf{A}$ is square and symmetric thus it admits the following eigendecomposition

\[
\mathbf{A} = \mathbf{D}^{-1} \times \mathbf{\Lambda} \times \mathbf{D}
\]

where $\mathbf{\Lambda}$ is a diagonal matrix storing the eigenvalues of $\mathbf{A}$ and $\mathbf{D}$ is an orthonormal matrix, whose columns coincide with the eigenvectors of $\mathbf{A}$.

In the light of the eigendecomposition of $\mathbf{A}$ we get:

\begin{dmath}
\norm{\left( \mathbf{I} -\delta\mathbf{A}\right)^{-1} - \mathbf{S}_k} 
= \norm{\delta^k\mathbf{A}^k + \delta^{k+1}\mathbf{A}^{k+1} + \ldots} 
\leq \norm{\delta^k\mathbf{A}^k} + \norm{\delta^{k+1}\mathbf{A}^{k+1}} + \ldots = \delta^k\norm{\mathbf{A}^k} + \delta^{k+1}\norm{\mathbf{A}^{k+1}} + \ldots
\end{dmath}

\noindent
which can be further simplified by observing that, for any $j$:

\begin{dmath}
{\norm{\mathbf{A}^j} = \norm{\mathbf{D}^{-1}\mathbf{\Lambda}\mathbf{D}} 
\leq \norm{\mathbf{D}^{-1}}\norm{\mathbf{\Lambda}^j}\norm{\mathbf{D}}}
= {\norm{\mathbf{\Lambda}^j} 
\leq  \norm{\mathbf{\Lambda}}^j  
= \lambda_1^j}
\end{dmath}

Here we used the fact that $\mathbf{D}$ and $\mathbf{D}^{-1}$ are orthonormal so their $L_2$ norm is equal to $1$. 
In addition, $\mathbf{\Lambda}$ has the same spectrum of $\mathbf{A}$ hence its $L_2$ norm coincides with the spectral radius $\lambda_1$ of $\mathbf{A}$.
By putting together the previous results we obtain:

\begin{dmath}
\norm{\left( \mathbf{I} -\delta\mathbf{A}\right)^{-1} - \mathbf{S}_k}  
\leq \delta^k\lambda_1^k + \delta^{k+1}\lambda_1^{k+1} + \ldots 
= \delta^k\lambda_1^k \left( 1 + \delta\lambda_1 + \delta^2\lambda_1^2 + \ldots\right) 
= \delta^k\lambda_1^k \frac{1}{1 - \delta\lambda_1} 
= \left(\delta \lambda_1\right)^k\frac{1}{1 - \delta\lambda_1}
\end{dmath}

\noindent
as required. 

\end{proof}

\subsection{Rate of convergence for the exponential potential gain}
\label{sub:rate-convergence-exponential}

The convergence result obtained with Theorem \ref{the:convergence-geometric} above has a counterpart for the exponential potential gain. 
In Theorem \ref{the:convergence-esponential} below we give an exponential convergence result for the exponential potential gain case.

\begin{theorem}
\label{the:convergence-esponential}
Let $\G$ be a graph with adjacency matrix $\A$ and let $\lambda_1$ be the spectral radius of $\mathbf{A}$. 
If $k$ is at least $2e \lambda_1$ then 

\[
\left(\frac{1}{2}\right)^{2 e \lambda_1} \lambda_1^{-\frac{1}{2}}
\]

\noindent
is an upper bound for $\varepsilon_e(k)$. 
\end{theorem}

\begin{proof}

Thanks to Equation \ref{eqn:exponential-potential-gain} we can define:

\[
\mathbf{S} = \sum_{i=1}^{+\infty} \frac{1}{i!} \mathbf{A}^i \quad \text{and} \quad \mathbf{S}_k = 
\sum_{i=1}^{k} \frac{1}{i!} \mathbf{A}^i
\]

\noindent
Next, we exploit the sub-multiplicativity property of the $L_2$ norm to obtain:

\begin{dmath}
\label{eq:error}
\norm{\mathbf{e} - \mathbf{e_k}} 
= \norm{\mathbf{A} \times \mathbf{S} \times \mathbf{1} - \mathbf{A} \times \mathbf{S_k}\times \mathbf{1}} 
= \norm{\mathbf{A} \times \left(\mathbf{S} -  \mathbf{S_k}\right)\times \mathbf{1}} 
\leq \norm{\mathbf{A}}\norm{\mathbf{R_k}}\norm{\mathbf{1}}
\end{dmath}

\noindent
with $\mathbf{R_k} = \mathbf{S} - \mathbf{S_k}$.
Also, by repeated application of the triangle inequality we obtain:

\begin{dmath}
\label{eqn:varepsilon}
\norm{\mathbf{R_k}} 
= \norm{\frac{\mathbf{A}^{k+1}}{\left(k+1\right)!} + \frac{\mathbf{A}^{k+2}}{\left(k+2\right)!} + \ldots}  
\leq \frac{\norm{\mathbf{A}^{k+1}}}{\left(k + 1\right)!} + \frac{\norm{\mathbf{A}^{k+2}}}{\left(k + 2\right)!} + \ldots
\end{dmath}

Recall that $\G$ is undirected so its adjacency matrix $\mathbf{A}$ is symmetric.
Thus, the $L_2$ norm of $\mathbf{A}$ coincides with its spectral radius. 
Also, for any $r \in \mathbb{N}$, $\mathbf{A}^r$ is still symmetric and, due to the sub-multiplicativity of the $L_2$ norm, we get $\norm{\mathbf{A}^r} \leq \lambda_1^r$, which allows us to simplify Equation \ref{eqn:varepsilon} as follows:

\begin{multline}
\norm{\mathbf{R_k}} \leq \frac{\lambda_1^{\left(k+1\right)}}{(k+1)!} +
\frac{\lambda_1^{\left(k+2\right)}}{(k+2)!} + \ldots =\\
= \frac{\lambda_1^{\left(k\right)}}{k!}\left[  1 + \frac{\lambda_1}{k+1} +  \frac{\lambda_1^2}{k(k+1)} +
\frac{\lambda_1^3}{k(k+1)(k+2)} + \ldots \right] \\\leq
\frac{\lambda_1^{\left( k \right)}}{k!}\left[  1 + \frac{\lambda_1}{k+1} +  \frac{\lambda_1^2}{(k+1)^2} +
\frac{\lambda_1^3}{(k+1)^3} + \ldots \right] = \\
= \frac{\lambda_1^{k}}{k!} \sum_{r=0}^{+\infty} \left(\frac{\lambda_1}{k+1}\right)^r
\end{multline}

Now, since $k + 1 > k \simeq 2e\lambda_1 > \lambda_1$, we have that $\sum_{r=0}^{+\infty} \left(\frac{\lambda_1}{k+1}\right)^r$ converges to the constant value $\frac{k+1}{k+1 - \lambda_1} \simeq \frac{2e}{2e - 1}$. 

The final step corresponds to applying Stirling's formula \cite{knuth1989concrete}, which states that, for sufficiently large values of $k$, $k! \simeq \sqrt{2\pi k} \left(\frac{k}{e}\right)^k$, which implies $k! = 2\sqrt{\pi} \lambda_1^{\frac{1}{2}} \left( 2\lambda_1\right)^{2e\lambda_1}$ if $k \simeq 2e \lambda_1$. 
Therefore, after some simplifications, we obtain

\[
\frac{\lambda_1^{k}}{k!} \simeq \frac{1}{2\sqrt{ \pi}}\left(\frac{1}{2}\right)^{2 e \lambda_1} \lambda_1^{-\frac{1}{2}}.
\]

\noindent
which completes the proof.

\end{proof}

\subsection{Computational Analysis}
\label{sub:computational}

\blue{Many previous studies focused on the problem of efficiently calculating the product $F(\mathbf{A}) \times \mathbf{b}$ where $\mathbf{A} \in \mathbb{R}^{n \times n}$, $\mathbf{b} \in \mathbb{R}^n$ and $F(\cdot)$ is an arbitrary function defined over the spectrum of $\mathbf{A}$ \cite{popolizio2008acceleration}.  \\
In many applications $\mathbf{A}$ is large and, thus, it is computationally prohibitive to first compute $F(\mathbf{A})$ and, then, to form the product $F(\mathbf{A}) \times \mathbf{b}$. \\
A clever strategy consists of projecting $\mathbf{A}$ (resp., $\mathbf{b}$) onto a matrix $\mathbf{H}$ (resp., a vector $\mathbf{w}$) of size $r\times r$ (resp., $r$) belonging to a subspace $\Omega$ such that $r$ is much smaller than $n$: in this way, we estimate the product $F(\mathbf{A}) \times \mathbf{b}$ as   
$F(\mathbf{H}) \times \mathbf{w}$ which is much easier to compute.\\
The task of projecting $\mathbf{A}$ and $\mathbf{b}$ onto $\mathbf{H}$ and $\mathbf{w}$ is equivalent to constructing an orthonormal matrix $\mathbf{V} = [\mathbf{v}_1, \mathbf{v}_2, \ldots, \mathbf{v}_r]$ whose columns span $\Omega$. 
If such a matrix $\mathbf{V}$ is available, then $\mathbf{H}$ is defined as
$\mathbf{H} = \mathbf{V}^T\mathbf{A}\mathbf{V}$ and the vector $\mathbf{b}$ is mapped onto the vector $\mathbf{w} = \mathbf{V}^T \times \mathbf{b}$.\\
The procedure for calculating $\mathbf{V}$ depends on both the spectral features of $\mathbf{A}$ as well as on the approximation accuracy we plan to obtain; in practice, the size of $\mathbf{V}$ could be very large and, thus, optimization techniques have been extensively studied to generate  good approximations of $F(\mathbf{A}) \times \mathbf{b}$.\\
One of the most popular techniques is the {\em Krylov subspace} \cite{saad1992analysis,heath2002scientific}
\begin{definition}[Standard Krylov Subspace]
Let $\mathbf{A} \in \mathbb{R}^{n \times n}$ matrix and let $\mathbf{b} \in \mathbb{R}^n$. The {\em standard Krylov subspace} of dimension $t$  associated with $\mathbf{A}$ and $\mathbf{b}$  is defined as the   
$$
\mathcal{K}_t(\mathbf{A}, \mathbf{b})=  \textsl{span} \{\mathbf{b}, \mathbf{A} \times \mathbf{b}, \ldots, \mathbf{A}^{t-1}\times \mathbf{b}\}
$$
\end{definition}
One could construct the standard Krylov subspace $\mathcal{K}_t(\mathbf{A}, \mathbf{b})$ and identify an orthonormal basis for $\mathcal{K}_t(\mathbf{A}, \mathbf{b})$; the vectors in such a basis constitute the columns of the matrix $\mathbf{V}$ we would like to compute.\\
Since $\mathbf{A}$ is symmetric the most efficient technique to find such an orthonormal basis for $\mathcal{K}_t(\mathbf{A}, \mathbf{b})$ is the {\em symmetric Lanczos algorithm} \cite{heath2002scientific}.
Such an algorithm performs $\ell$ iterations and, at each iteration, the most time-expensive step consists of calculating the product of the matrix $\mathbf{A}$ by a vector of size $n$. 
If $\mathbf{A}$ is large but sparse, the cost of each iteration is $O(n)$ and the overall cost is $O(\ell n)$ which is equal to the asymptotic cost of our algorithm.\\ 
A further, computationally-appealing option is to approximate  $F(\cdot)$ by means of a rational function $g(x)$ defined as the ratio of two polynomials, $p_{\nu}(x)$ and $q_{\mu}(x)$, of degree $\nu$ and $\mu$, respectively \cite{guttel2013rational,popolizio2008acceleration}.
In this way, the problem of calculating $F(\mathbf{A}) \times \mathbf{b}$ is equivalent to evaluate $g(\mathbf{A}) \times \mathbf{b}$ which should be hopefully easier to calculate.
If we denote as $\sigma_1, \sigma_2, \ldots \sigma_{\mu}$ the poles of $g(x)$, i.e, the roots of $q_{\mu}(x)$, then we introduce the following definition:
\\
\begin{definition}[Rational Krylov Subspace]
Let $\mathbf{A} \in \mathbb{R}^{n \times n}$ matrix and let $\mathbf{b} \in \mathbb{R}^n$. The {\em rational Krylov subspace} of dimension $t$  associated with $\mathbf{A}$ and $\mathbf{b}$  is defined as the   
$$
\mathcal{K}_t^{\star}(\mathbf{A}, \mathbf{b}) = \textsl{span}\{\mathbf{b}, (\mathbf{A} - \sigma_1 \mathbf{I})^{-1} \times \mathbf{b}, \ldots, (\mathbf{A} - \sigma_t \mathbf{I})^{-1} \times \mathbf{b}\}
$$
\end{definition}
Rational Krylov subspace methods are often more accurate than standard ones but they require to solve, at each iteration, a linear system; such a task is equivalent to computing terms of the form $(\mathbf{A} - \sigma_i \mathbf{I})^{-1} \times \mathbf{b}$ and such a procedure is usually much more expensive than standard Krylov methods.\\
Theorems \ref{the:convergence-geometric} and \ref{the:convergence-esponential}  indicate that our procedure is capable of achieving an exponentially decay error: consequently, we are able to achieve the accuracy we wish with a relatively small number of iterations.}

\section{Experimental validation}
\label{sec:experiments}

In this section we report on the experiments we carried out to assess the effectiveness of geometric and exponential potential gain on real-world datasets.
Our experiments aim at answering the following questions:

\begin{enumerate}
\item[$Q_1:$] How sensitive are our approximations of the geometric and the exponential potential gain w.r.t. the length $k$ of walks?

\item[$Q_2:$] Do our algorithms scale up to real graphs?

\item[$Q_3:$] How do the geometric and exponential potential gain correlate with other, popular, centrality metrics such as Degree, Katz, PageRank and Eigenvector Centrality? \blue{Are the aforementioned centrality metrics good candidates for assessing the navigability of a node in the sense defined in this paper?}
\end{enumerate}

To answer these questions, we considered three large, real datasets, taken from~\cite{kunegis2013konect}, whose features are described in 
Table~\ref{tbl:dataset}.
\blue{The first dataset -- \textsc{Facebook} -- is a sample of the Facebook user connections graph: in it, a node represents a Facebook user and an edge represents a friendship between two users. The second dataset, called \textsc{DBLP}, is a sample of the DBLP computer science bibliography: authors correspond to nodes and two nodes are linked by an edge if the corresponding authors have published at least one paper together. Finally, \textsc{YouTube}, is a sample of the friendship network between YouTube users.}

We implemented our algorithms in Python (with the Scipy module) on a hardware platform with the following features: AMD Ryzen 5 1600 CPU, 16GB RAM and Ubuntu 17.10.   
 
\begin{table}
\centering
\begin{tabular}{clrrrr}
\hline \hline 
& Dataset & Nodes & Edges & {$\lambda_1$} & Year\\
\hline 
1 & Facebook Friendship & $63\,731$ & $817\,035$ & $132.57$ & 2009  \\
2 & DBLP co-authorship & $317\,080$ & $1\,049\,866$ & $115.85$ &  2012 \\
3 & Youtube friendship & $1\,134\,890$ & $2\,987\,624$ & $210.40$ & 2012 \\
\hline
\end{tabular}
\caption{Key features of the datasets employed in our experimental tests. For each dataset we report the number of nodes, the number of edges, the spectral radius and the year in which data were collected.}
~\label{tbl:dataset}
\end{table}

\subsection{The impact of walk length on the approximation of geometric/exponential potential gains}
\label{sub:accuracy-estimation}

The aim of this section is to answer question $\mathbf{Q_1}$ and, specifically, to study the quality of the approximation of the geometric and exponential potential gain in relation to walk length. 
Ideally, one would like geometric (resp., exponential) potential gain values to stabilize already for small values of $k$.
It would mean that, for any node $i$, it suffices to consider nodes located a few hops away from $i$ to get a satisfactory approximation of its geometric (resp., exponential) potential gain.

To perform our study we experimentally studied the decrease of $\varepsilon_g(k)$ and $\varepsilon_e(k)$ as function of $k$.
We start by discussing the results for the geometric potential gain, in Figures \ref{fig:error-facebook}, \ref{fig:error-dblp} and \ref{fig:error-youtube}.

\begin{figure}[!t]
\centering
\includegraphics[width=.6\columnwidth]{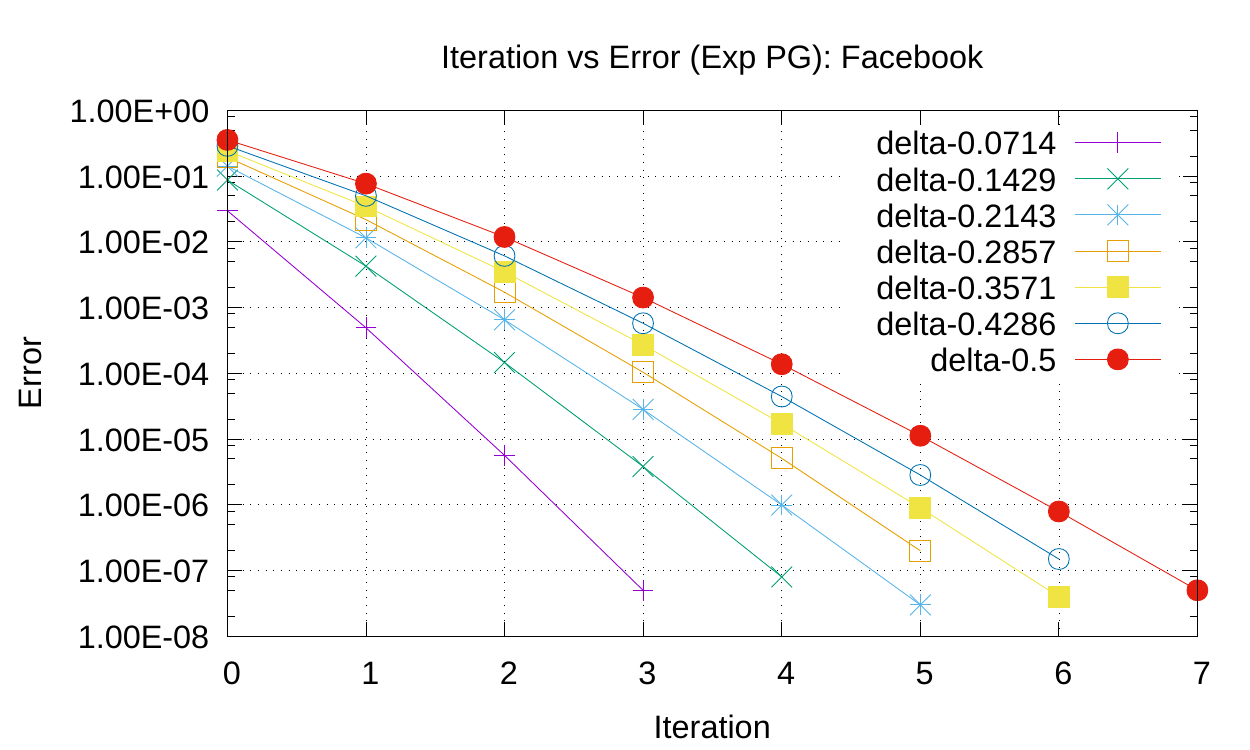}
\caption{Values of $\varepsilon_g(k)$ as a function of $k$ plotted in semi-logarithmic scale for \textsc{Facebook}}
\label{fig:error-facebook}
\end{figure}

\begin{figure}[!t]
\centering
\includegraphics[width=.6\columnwidth]{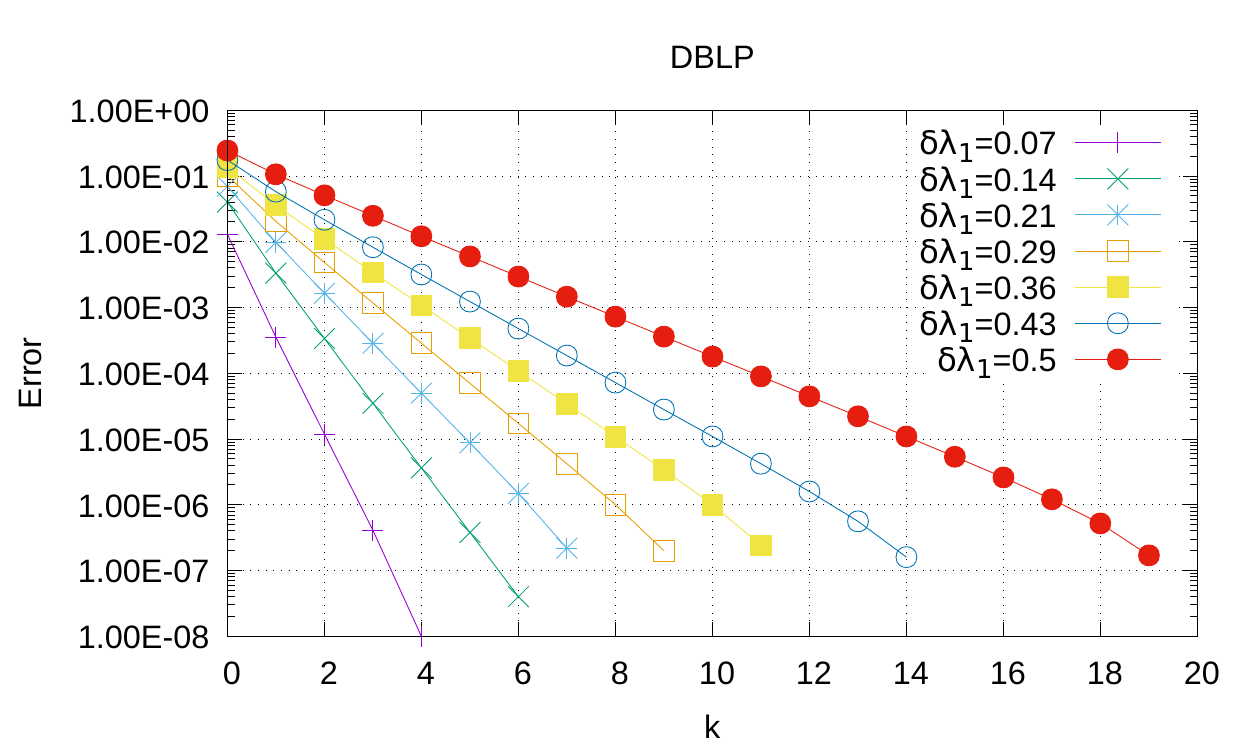}
\caption{Values of $\varepsilon_g(k)$ as a function of $k$ plotted in semi-logarithmic scale for \textsc{DBLP}}
\label{fig:error-dblp}
\end{figure}

\begin{figure}[!t]
\centering
\includegraphics[width=.6\columnwidth]{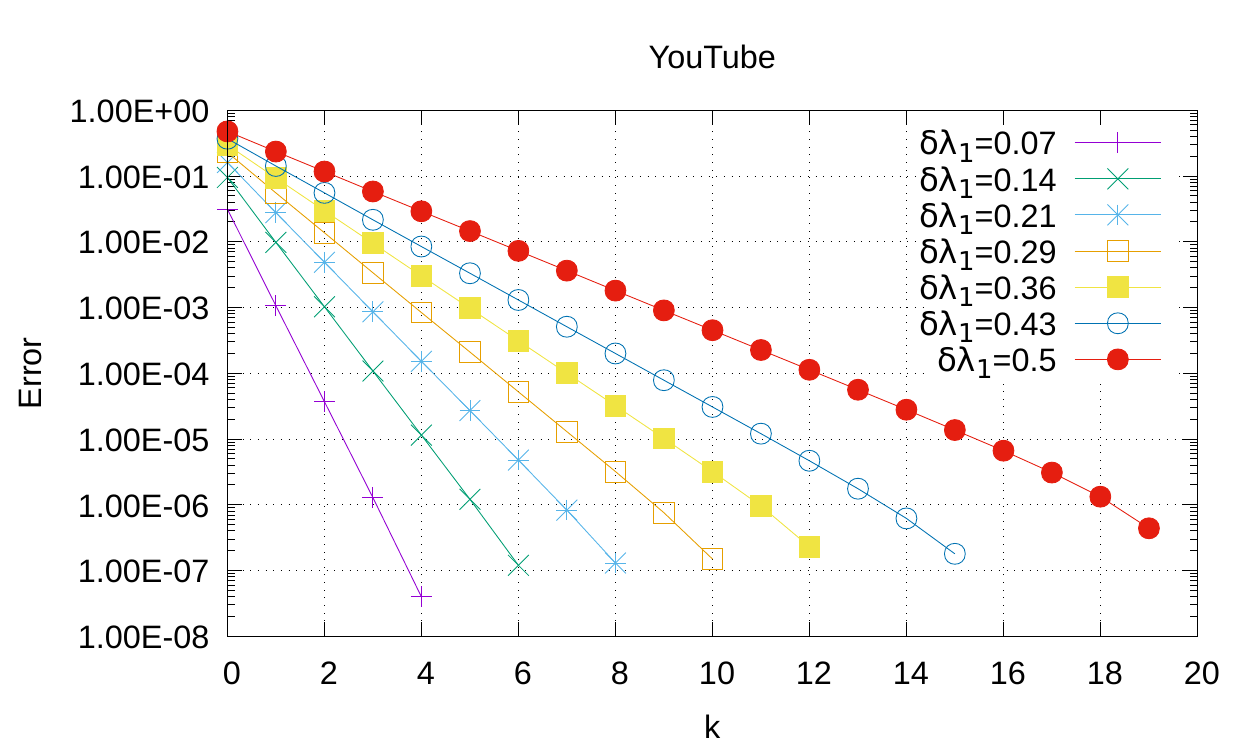}
\caption{Values of $\varepsilon_g(k)$ as a function of $k$ plotted in semi-logarithmic scale for \textsc{YouTube}}
\label{fig:error-youtube}
\end{figure}

A notable feature of our algorithm is that convergence is very fast, independently of the size and nature of the dataset under investigation.
For instance, with the \textsc{YouTube} dataset-- the largest tested here -- walks up to length $k = 20$ are sufficient to achieve $\varepsilon_g(k)$ lesser than $10^{-6}$.

\blue{A further observation is that, for a fixed value of $k$, the larger $\delta$, the slower the convergence of $\varepsilon_g(k)$ to zero, and that our results perfectly agree with the statement of Theorem \ref{the:convergence-geometric}.\\
It is also instructive to study how $\varepsilon_g(k)$ varies across datasets: for any value of $\delta$, $\varepsilon_g(k)$ for \textsc{DBLP} converges to zero faster than it does for \textsc{Facebook}, despite the fact that \textsc{DBLP} is about five times larger than \textsc{Facebook} (see Figures \ref{fig:error-dblp} and \ref{fig:error-youtube}).}
It is also possible to appreciate small differences in the slopes of straight lines plotting $\log \varepsilon_g(k)$ as a function of $k$.

The important finding described above is mirrored by a similar result for the exponential potential gain, as illustrated by Figure \ref{plt:errEpg}.
Once again we notice that graph size has a small impact on the convergence of our algorithm.
From Theorem \ref{the:convergence-geometric}, in fact, $\lambda_1$ is the dominant parameter: the smaller $\lambda_1$, the faster the algorithm converges to the true value of the exponential potential gain. 
Notice also how the exponential potential gain needs walks that are \textit{longer} than those needed for the geometric potential gain: we need walks up to length 169, 189 and 279 for \textsc{DBLP}, \textsc{Facebook} and \textsc{YouTube}, respectively.  

\begin{figure}[htb]
\centering
\includegraphics[width=.7\columnwidth]{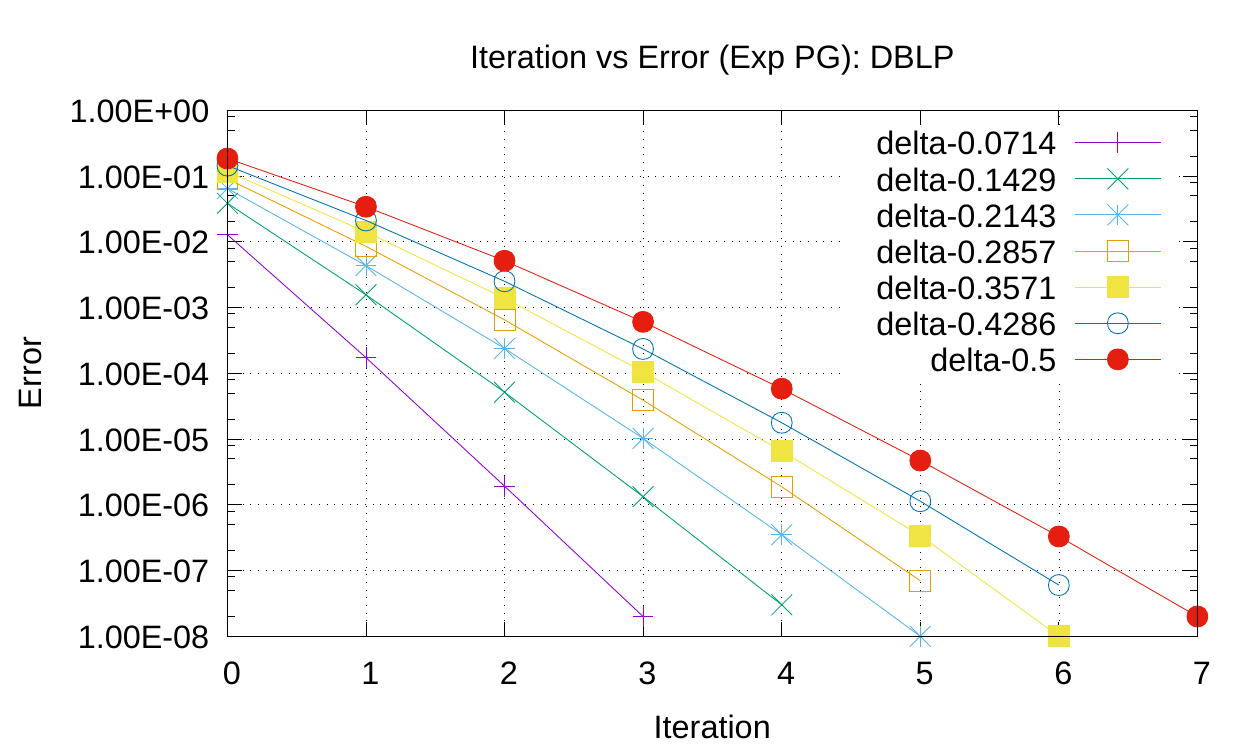}
\caption{Values of $\varepsilon_e(k)$ as a function of $k$ for the \textsc{Facebook}, \textsc{DBLP}, and \textsc{YouTube} datasets plotted in semi-logarithmic scale.}
\label{plt:errEpg}
\end{figure}

The main results of our experimental validation can be now summarized as follows: {\em (i)} For the geometric potential gain, small values of $\delta$ should be used.  Walks of length between 4 and 10 are sufficient to get a very good approximation. {\em (ii)} The time needed to compute the geometric and the exponential potential gain does not depend on the graph size but it only depends on $\lambda_1$: the larger $\lambda_1$ the more dense/connected the graph is and, thus, a larger number of walks is needed to achieve a good approximation. {\em (iii)} Computing the exponential potential gain is slower than computing geometric potential gain and, experimentally, it might require walks whose length is ten times larger.  This finding suggests that we should consider as future work to introduce a decay factor of the form $\frac{\delta^k}{k!}$ to penalize long walks.

\subsection{Scalability Analysis}
\label{sub:scalability-analysis}

In this section we address question $\mathbf{Q_2}$ by studying how our algorithms scale over \textsc{Facebook}, \textsc{DBLP} and \textsc{YouTube}. 
We measured the execution times as a function of the walk length $k$.
In particular, for the geometric potential gain we were concerned with understanding how $\delta$ affected the performance of our approach.
The results we obtained are plotted in Figure \ref{plt:times}.

Clearly, an increase of $\delta$ yields an increase in computational time. 
This is due to the fact that larger values of $\delta$ force our algorithm to explore the graph in more depth. 
Such effect can be clearly seen in Figure \ref{plt:times}; the increase is approximately linear in $\delta$ so we can conclude that the computational impact of increasing $\delta$ is limited.
Of course, larger datasets still require greater computational resources since we will have multiply the adjacency matrix $\mathbf{A}$ by a vector.
However, for sparse adjacency matrices the calculation is still fast, e.g., it takes less than one second for the three datasets considered here.

\begin{figure}
\centering
\includegraphics[width=.7\columnwidth]{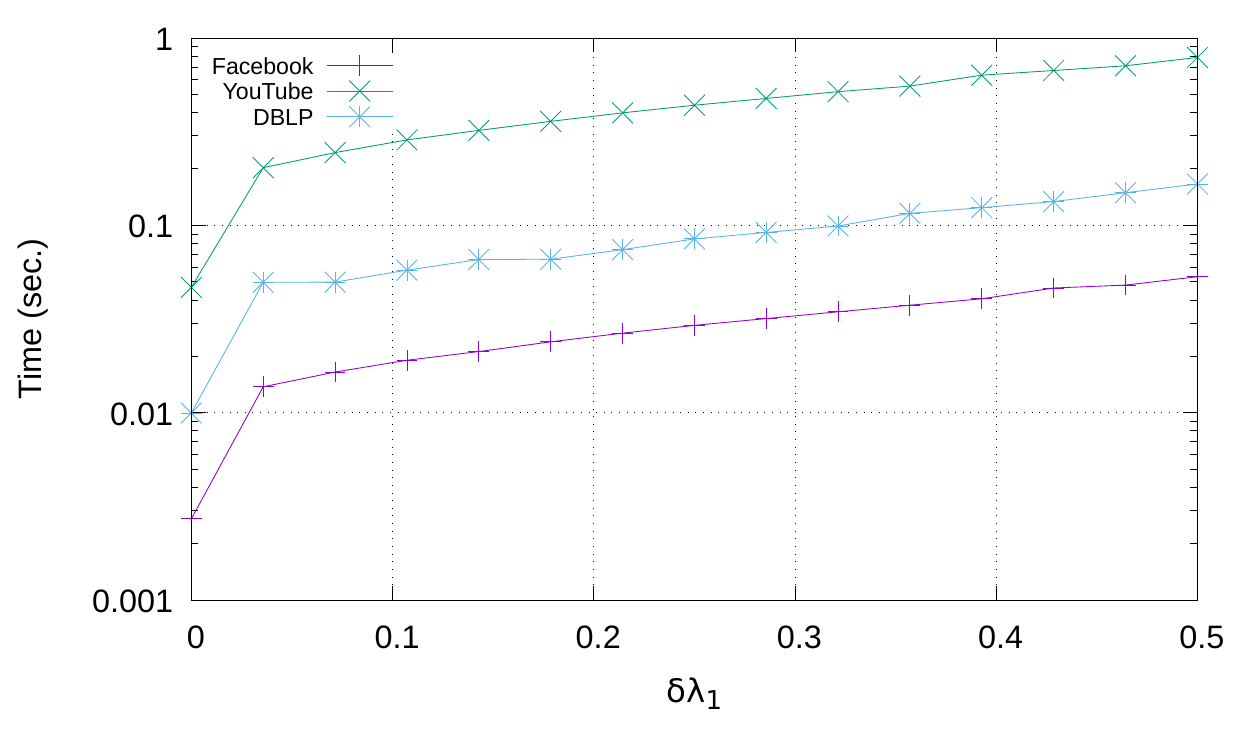}
\caption{Computation times vs. $\delta \lambda_1$ for the geometric potential gain}
\label{plt:times}
\end{figure}

%

Table \ref{tbl:comp-times-epg} reports the computation times of the exponential potential gain. 
Again we may notice how exponential potential gain is computationally more demanding than the  geometric potential gain as we need many more iterations. 
In addition, notice how the computational time for DBLP is about three times slower than for \textsc{Facebook} despite the fact that the latter needed 20 iterations more.
Such difference depends on the difference in size between the two datasets.

\begin{table}[htb]
\centering
\begin{tabular}{lr}
\hline \hline 
Dataset &  Time (sec.) \\
\hline
Facebook & $0.525$ \\
DBLP &  $1.568$ \\ 
YouTube & $11.878$ \\ 
\hline \hline 
\end{tabular}
\caption{Computational times for the exponential potential gain}~\label{tbl:comp-times-epg}
\end{table}

\subsection{Relation with other centrality metrics}
\label{sub:correlation-analysis}

In this section we investigate how the geometric and the exponential potential gain are correlated with some popular centrality metrics. 
Specifically, we compared the Geometric (\emph{GPG}) and the Exponential Potential Gain (\emph{EPG}) with Degree Centrality \emph{(DEG)}, Eigenvector Centrality \emph{(EC),} PageRank \emph{(PR)} and Katz Centrality \emph{(Katz).}
As for \emph{PR,} we fixed the damping factor equal to 0.85. 
For a fair comparison, we choose the same value of $\delta$ for\emph{GPG} and \emph{Katz}.
We used Spearman's $\rho$ coefficient to calculate the correlation between the ranked lists of nodes that each of these centrality metrics generates on the datasets above.

\begin{figure}[!t]
\centering
\includegraphics[width=.7\columnwidth]{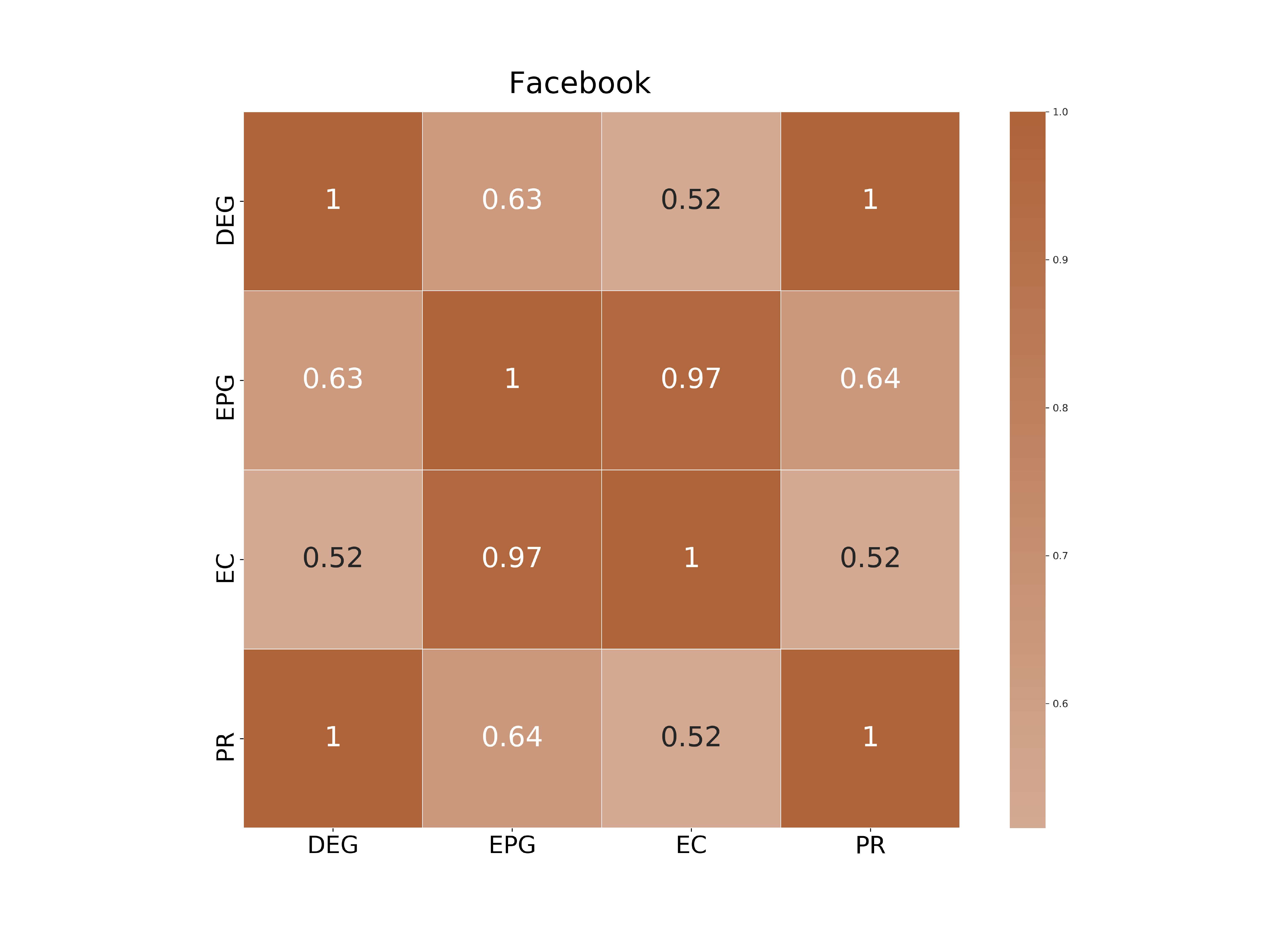}
\caption{Spearman's $\rho$ correlation between \textit{DEG, EPG, EC} and \textit{PR} on the \textsc{Facebook} dataset. }
\label{fig:corr-facebook-heatmap}
\end{figure}

\begin{figure}[!t]
\centering
\includegraphics[width=.7\columnwidth]{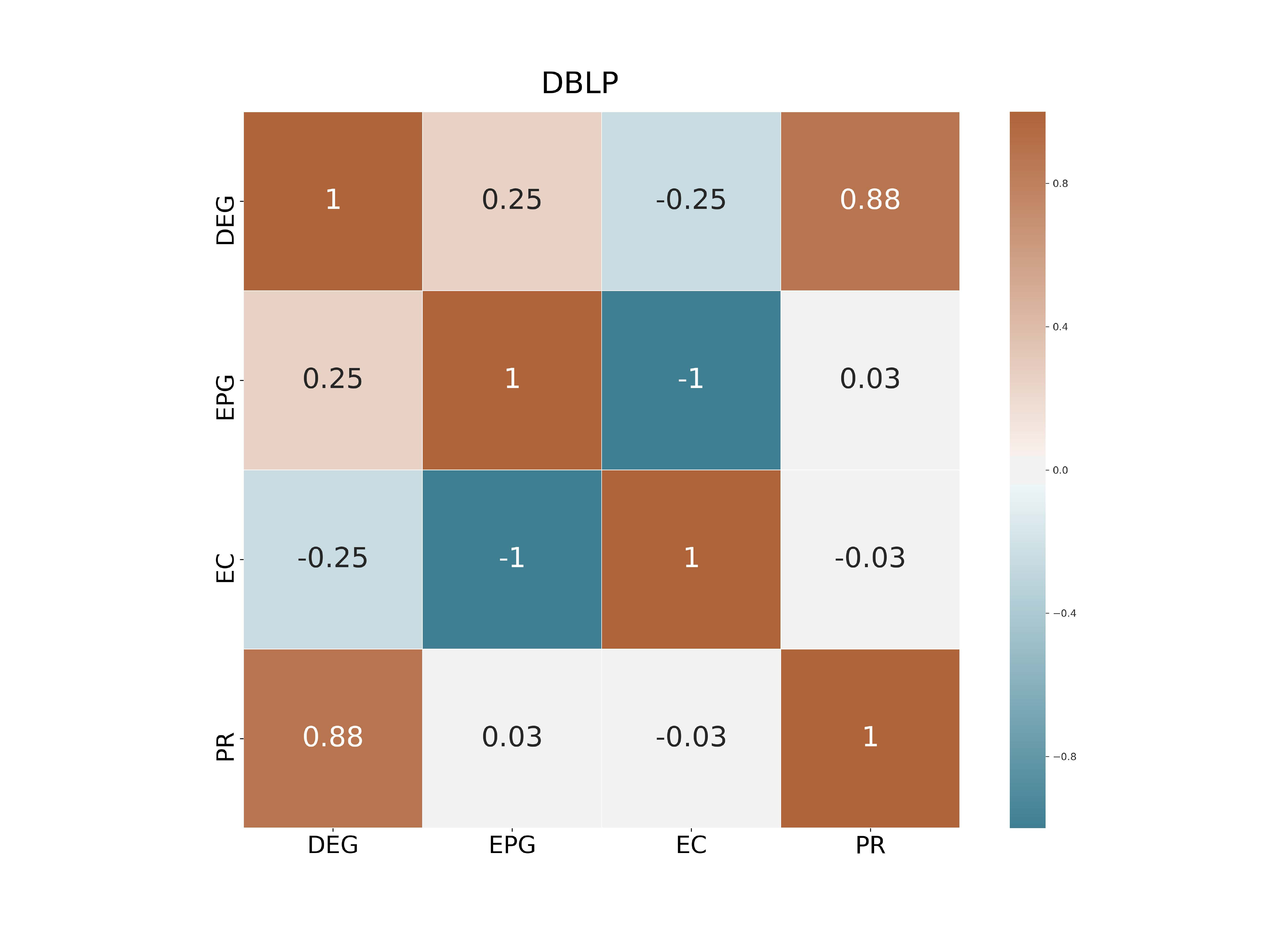}
\caption{Spearman's $\rho$ correlation between \textit{DEG, EPG, EC} and \textit{PR} on the \textsc{DBLP} dataset. }
\label{fig:corr-dblp-heatmap}
\end{figure}

\begin{figure}[!t]
\centering
\includegraphics[width=.7\columnwidth]{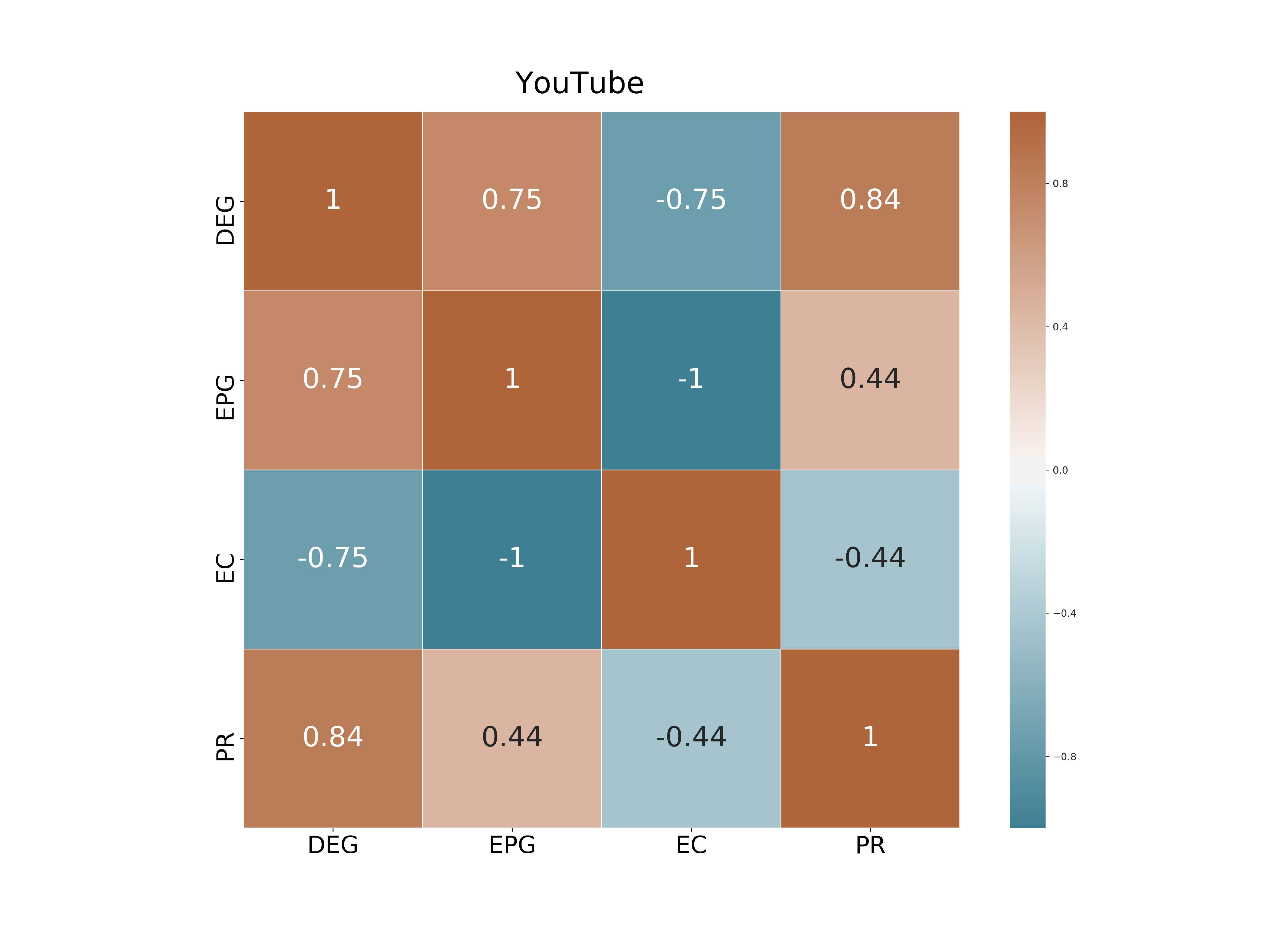}
\caption{Spearman's $\rho$ correlation between \textit{DEG, EPG, EC} and \textit{PR} on the \textsc{YouTube} dataset. }
\label{fig:corr-youtube-heatmap}
\end{figure}


Let us now analyse how \emph{DEG,} \emph{PR} and \emph{EC} correlate each other and how they are related with \emph{EPG.} (\emph{EPG} is discussed later as it does not depend on any parameter).
The main findings of our analysis are plotted in Figures \ref{fig:corr-facebook-heatmap} - \ref{fig:corr-youtube-heatmap} and can be summarized as follows:

\begin{enumerate}
\item \emph{DEG} and \emph{PR} are strongly positively correlated: we found Spearman's $\rho$ correlation coefficients ranging from $0.84$ (on \textsc{YouTube}) to $1$ (on \textsc{Facebook}). 
Such correlation is well-known in the literature because, in the case of undirected graphs, the {\em PR} scores of nodes are almost proportional to their degrees.

\item On \textsc{Facebook} \emph{EPG} and \emph{EC} are highly correlated while at the same time display a perfect negative correlation on the other two datasets.
This result is related to the property that the \emph{EPG} of a node is always proportional to its communicability, as we have shown in Section \ref{sub:geometric-exponential-potantial-gain}. 
When the spectral gap between the largest and the second-largest eigenvalue, i.e., $\lambda_1 - \lambda_2$, is large, then communicability is, up to a constant factor, exactly equal to \emph{EC;} see \cite{benzi2013total} for a detailed proof.
As a consequence, for graphs with a large spectral gap the \emph{EPG} of a node is proportional to its \emph{EC}; this explaining the large correlation values we observed in our tests.

\item The topological features of the given graph have a huge impact on $\rho$: for instance, \emph{DEG} and \emph{EC} are positively correlated on \textsc{Facebook} ($\rho = 0.52$) but they are negatively correlated on \textsc{DBLP} and \textsc{YouTube.}
Analogously, \emph{EPG} and \emph{PR} are positively correlated on \textsc{Facebook} and \textsc{YouTube} ($\rho$ = 0.64 and 0.44, respectively) but almost unrelated on \textsc{DBLP}.    
\end{enumerate}

Let us now consider the correlation between {\em DEG}, {\em EC}, {\em PR} and {\em EPG} and {\em GPG}; such correlation varies upon different choices of the parameter $\delta$.
Figures \ref{fig:corr-facebook-alpha}, \ref{fig:corr-dblp-alpha} and \ref{fig:corr-youtube-alpha}, show how $\rho$ varies as function of $\delta^{\star} = \delta / \delta_{\max}$, where $\delta_{\max} = \frac{1}{\lambda_1}$. 
We can draw the following conclusions:

\begin{enumerate}
\item The {\em GPG} and {\em Katz} have an almost perfect correlation, as one would expect (red line). 
In fact, {\em GPG} and {\em Katz} differ by a multiplicative and constant factor given by the degree.

\item We note some interesting facts about the correlation of {\em GPG} and {\em EPG} (orange line). 
Firstly, if $\delta^{\star}$ increases, then we generally have a (sometimes slight) increase in $\rho$.
The correlation coefficient $\rho$ is always positive but it ranges from $0.45$ in \textsc{DBLP} to 1 (on \textsc{Facebook} and \textsc{YouTube}).\\ 
To explain the differences emerging across our datasets we observe the difference of the geometric and the exponential potential gain of a node $i$ depends on $\frac{\lambda_i}{1-\lambda_i}- \lambda_i e^{\lambda_i}$ (see Equation \ref{eqn:difference-geometric-potential}). 
Therefore, on the basis of the distribution of eigenvalues, we conclude that for some datasets {\em GPG} and {\em EPG} might be strongly and positively correlated while they are negatively correlated elsewhere.

\item As for the correlation between {\em EC} and {\em GPG} (green line), we find a positive correlation on \textsc{Facebook} and a negative correlation on \textsc{DBLP} and \textsc{YouTube}.
In fact, for a fixed node, the {\em GPG} is proportional to the {\em Katz} score of that node; if $\delta^{\star} \to 1$, then the {\em Katz} score converges to {\em EC} and, in this case, the {\em GPG} of a node is proportional to its {\em EC}.
Similar considerations hold for {\em GPG} and {\em PR} (magenta line).

\item Consider now {\em DEG} (see blue line).
An increase in $\delta^{\star}$ generally causes a decrease in $\rho$, with the exception of the \textsc{DBLP} dataset where the correlation between {\em DEG} and {\em GPG} is almost perfect. 
As expected, when $\delta^{\star} \to 0$ the {\em GPG} is well approximated by {\em DEG}. 
Vice versa, as $\delta^{\star}$ increases, walks of length greater than one increasingly contribute to {\em GPG}, thus amplifying the difference between {\em DEG} and {\em GPG}.
\end{enumerate}

\begin{figure}[!t]
\centering
\includegraphics[width=.7\columnwidth]{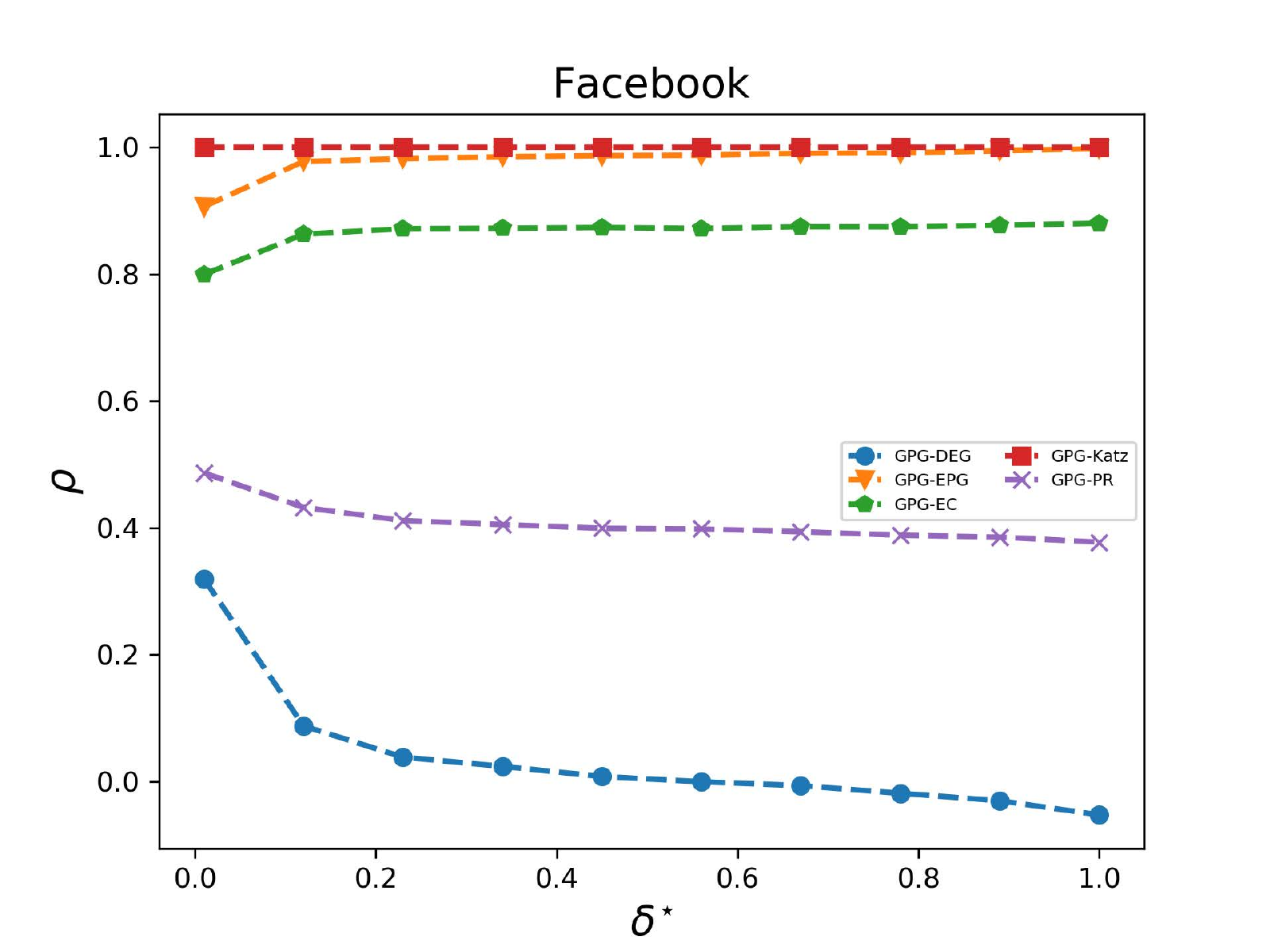}
\caption{Spearman's $\rho$ correlation between GPG and DEG, EC, PR, Katz as function of $\delta^{\star}$ on the \textsc{Facebook} dataset.}
\label{fig:corr-facebook-alpha}
\end{figure}

\begin{figure}[!t]
\centering
\includegraphics[width=.7\columnwidth]{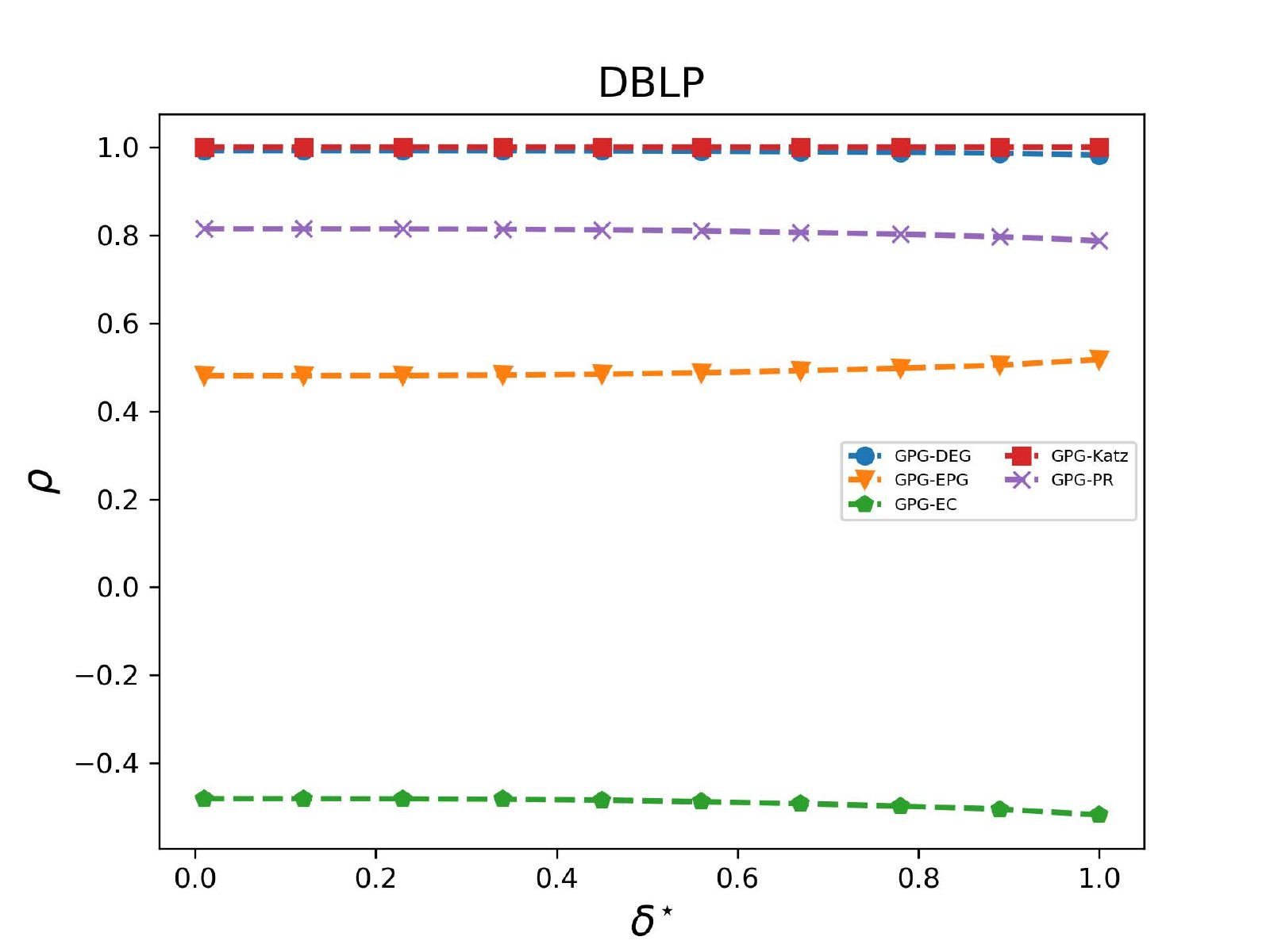}
\caption{Spearman's $\rho$ correlation between GPG and DEG, EC, PR, Katz as function of $\delta^{\star}$ on the \textsc{DBLP} dataset.}
\label{fig:corr-dblp-alpha}
\end{figure}

\begin{figure}[!t]
\centering
\includegraphics[width=.7\columnwidth]{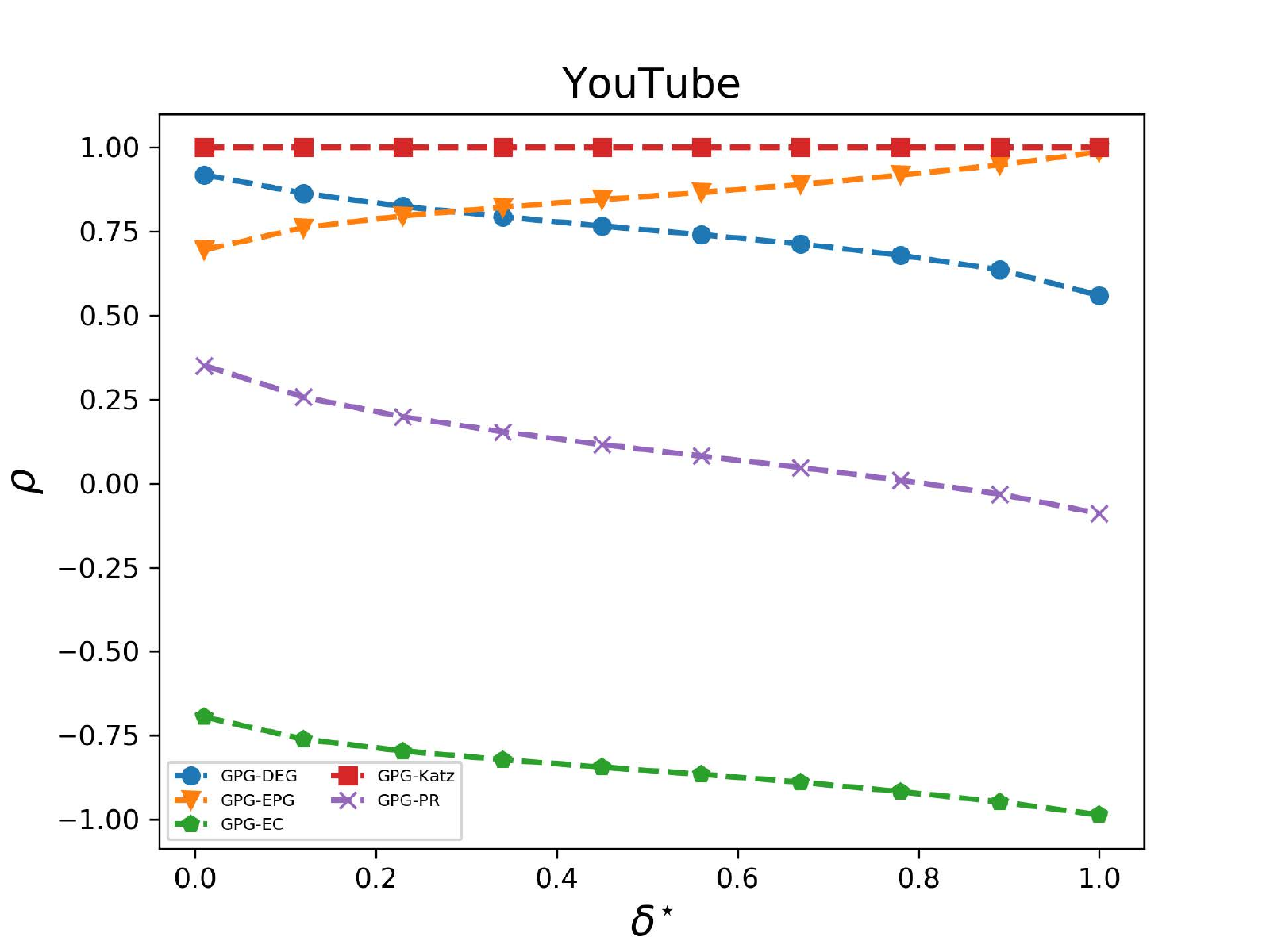}
\caption{Spearman's $\rho$ correlation between GPG and DEG, EC, PR, Katz as function of $\delta^{\star}$ on the \textsc{YouTube} dataset.}
\label{fig:corr-youtube-alpha}
\end{figure}

\noindent
The detailed analysis of the experimental results strongly suggests that the potential gain is a general framework, which unifies some of the main walk-based centralities.

\section{Conclusions}\label{sec:conc}

We have introduced the potential gain, a general framework which captures the ability of a node to act as a target point for navigation within a graph and unifies several walk-based centrality indices in graphs.
We have defined two variants of the potential gain, the geometric and exponential potential gain and proposed two iterative algorithms for each of them. The convergence of the algorithms was also proved;
their scalability was tested against three real large datasets. 

%
%
%
%



The present results provide a starting point to a better understanding of the theoretical properties of the potential gain.
Previous research has sought to axiomatise centrality and its metrics, and it would be extremely interesting to study which of these axioms are satisfied by the Potential gain. 
For instance, Boldi and Vigna \cite{boldi2014axioms} proposed three axioms, called {\em size,} {\em density} and {\em score monotonicity,} which may play a role in axiomatising the Potential Gain itself.

%

Another topic for future work is investigation of the relationship between network robustness and network navigability. 
To this end, we intend to design an experiment in which graph nodes are ranked on the basis of their geometric/exponential potential gain and then are progressively removed from the graph.
Basic properties about graph topology, such as the number and size of connected components shall be re-evaluated upon node deletion.
We also plan to study how adding edges can increase the geometric/exponential potential gain of a target group of nodes.

\ifCLASSOPTIONcompsoc
%
%

\ifCLASSOPTIONcaptionsoff
  \newpage
\fi



%
%
%
\bibliographystyle{IEEEtran}
\bibliography{potential_gain}

\begin{thebibliography}{10}
\providecommand{\url}[1]{#1}
\csname url@samestyle\endcsname
\providecommand{\newblock}{\relax}
\providecommand{\bibinfo}[2]{#2}
\providecommand{\BIBentrySTDinterwordspacing}{\spaceskip=0pt\relax}
\providecommand{\BIBentryALTinterwordstretchfactor}{4}
\providecommand{\BIBentryALTinterwordspacing}{\spaceskip=\fontdimen2\font plus
\BIBentryALTinterwordstretchfactor\fontdimen3\font minus
  \fontdimen4\font\relax}
\providecommand{\BIBforeignlanguage}[2]{{%
\expandafter\ifx\csname l@#1\endcsname\relax
\typeout{** WARNING: IEEEtran.bst: No hyphenation pattern has been}%
\typeout{** loaded for the language `#1'. Using the pattern for}%
\typeout{** the default language instead.}%
\else
\language=\csname l@#1\endcsname
\fi
#2}}
\providecommand{\BIBdecl}{\relax}
\BIBdecl

\bibitem{lu2016vital}
L.~L{\"u}, D.~Chen, X.~Ren, Q.~Zhang, Y.~Zhang, and T.~Zhou, ``Vital nodes
  identification in complex networks,'' \emph{Physics Reports}, vol. 650, pp.
  1--63, 2016.

\bibitem{chung2009distributing}
F.~Chung, P.~Horn, and A.~Tsiatas, ``{Distributing antidote using Pagerank
  vectors},'' \emph{Internet Mathematics}, vol.~6, no.~2, pp. 237--254, 2009.

\bibitem{leskovec2007dynamics}
J.~Leskovec, L.~Adamic, and B.~Huberman, ``The dynamics of viral marketing,''
  \emph{ACM Transactions on the Web}, vol.~1, no.~1, p.~5, 2007.

\bibitem{AgresteMFPP15}
\BIBentryALTinterwordspacing
S.~Agreste, P.~{De Meo}, E.~Ferrara, S.~Piccolo, and A.~Provetti, ``{Trust
  Networks: Topology, Dynamics, and Measurements},'' \emph{{IEEE} Internet
  Computing}, vol.~19, no.~6, pp. 26--35, 2015. [Online]. Available:
  \url{https://doi.org/10.1109/MIC.2015.93}
\BIBentrySTDinterwordspacing

\bibitem{albert2004structural}
R.~Alber, I.~Albert, and G.~Nakarado, ``{Structural vulnerability of the North
  American power grid},'' \emph{Physical review E}, vol.~69, no.~2, p. 025103,
  2004.

\bibitem{Agreste-anobii15}
S.~Agreste, P.~{De Meo}, E.~Ferrara, S.~Piccolo, and A.~Provetti, ``Analysis of
  a heterogeneous social network of humans and cultural objects,'' \emph{{IEEE}
  Transactions on Systems, Man, and Cybernetics: Systems}, vol.~45, no.~4, pp.
  559--570, 2015.

\bibitem{newman2010networks}
M.~Newman, \emph{Networks: an introduction}.\hskip 1em plus 0.5em minus
  0.4em\relax Oxford University Press, 2010.

\bibitem{boldi2014axioms}
P.~Boldi and S.~Vigna, ``Axioms for centrality,'' \emph{Internet Mathematics},
  vol.~10, no. 3-4, pp. 222--262, 2014.

\bibitem{katz1953new}
L.~Katz, ``A new status index derived from sociometric analysis,''
  \emph{Psychometrika}, vol.~18, no.~1, pp. 39--43, 1953.

\bibitem{benzi2014matrix}
M.~Benzi and C.~Klymko, ``On the limiting behavior of parameter-dependent
  network centrality measures,'' \emph{SIAM Journal on Matrix Analysis and
  Applications}, vol.~36, no.~2, pp. 686--706, 2015.

\bibitem{brin1998anatomy}
S.~Brin and L.~Page, ``{The anatomy of a large-scale hypertextual Web search
  engine},'' \emph{Computer networks and ISDN systems}, vol.~30, no. 1-7, pp.
  107--117, 1998.

\bibitem{travers1967small}
J.~Travers and S.~Milgram, ``The small world problem,'' \emph{Phychology
  Today}, vol.~1, no.~1, pp. 61--67, 1967.

\bibitem{kleinberg2000small}
J.~Kleinberg, ``The small-world phenomenon: An algorithmic perspective,'' in
  \emph{Proc.\ of the ACM symposium on Theory of computing (STOC 2000)}.\hskip
  1em plus 0.5em minus 0.4em\relax ACM, 2000, pp. 163--170.

\bibitem{GoMuWa09}
S.~Goel, R.~Muhamad, and D.~Watts, ``Social search in \lq\lq small-world\rq\rq
  experiments,'' in \emph{Proc.\ of the International Conference on World Wide
  Web ( {WWW} 2009)}, Madrid, Spain, 2009, pp. 701--710.

\bibitem{LeHo08}
J.~Leskovec and E.~Horvitz, ``Planetary-scale views on a large
  instant-messaging network,'' in \emph{Proc.\ of the International Conference
  on World Wide Web, {WWW} 2008}, Beijing, China, 2008, pp. 915--924.

\bibitem{jeong2000large}
H.~Jeong, B.~Tombor, R.~Albert, Z.~Oltvai, and A.~Barabasi, ``The large-scale
  organization of metabolic networks,'' \emph{Nature}, vol. 407, no. 6804, p.
  651, 2000.

\bibitem{broder2000graph}
A.~Broder, R.~Kumar, F.~Maghoul, P.~Raghavan, S.~Rajagopalan, R.~Stata,
  A.~Tomkins, and J.~Wiener, ``{Graph structure in the Web},'' \emph{Computer
  Networks}, vol.~33, no. 1-6, pp. 309--320, 2000.

\bibitem{newman2001structure}
M.~Newman, ``The structure of scientific collaboration networks,''
  \emph{Proceedings of the National Academy of Sciences}, vol.~98, no.~2, pp.
  404--409, 2001.

\bibitem{dodds2003experimental}
P.~Dodds, R.~Muhamad, and D.~Watts, ``An experimental study of search in global
  social networks,'' \emph{Science}, vol. 301, no. 5634, pp. 827--829, 2003.

\bibitem{watts1998collective}
D.~Watts and S.~Strogatz, ``Collective dynamics of small-world networks,''
  \emph{Nature}, vol. 393, no. 6684, p. 440, 1998.

\bibitem{fenner2008modelling}
T.~Fenner, M.~Levene, and G.~Loizou, ``{Modelling the navigation potential of a
  Web page},'' \emph{Theoretical Computer Science}, vol. 396, no. 1-3, pp.
  88--96, 2008.

\bibitem{benzi2013total}
M.~Benzi and C.~Klymko, ``Total communicability as a centrality measure,''
  \emph{Journal of Complex Networks}, vol.~1, no.~2, pp. 124--149, 2013.

\bibitem{estrada2005subgraph}
E.~Estrada and J.~Rodriguez-Velazquez, ``Subgraph centrality in complex
  networks,'' \emph{Physical Review E}, vol.~71, no.~5, p. 056103, 2005.

\bibitem{csimcsek2008navigating}
O.~Simsek and D.~Jensen, ``Navigating networks by using homophily and degree,''
  \emph{Proceedings of the National Academy of Sciences}, vol. 105, no.~35, pp.
  12\,758--12\,762, 2008.

\bibitem{west2009wikispeedia}
R.~West, J.~Pineau, and D.~Precup, ``{Wikispeedia: An Online Game for Inferring
  Semantic Distances between Concepts},'' in \emph{Proc.\ of the International
  Joint Conference on Artificial Intelligence (IJCAI 2009)}, Pasadena,
  California, USA, 2009, pp. 1598--1603.

\bibitem{west2012automatic}
R.~West and J.~Leskovec, ``Automatic versus human navigation in information
  networks.'' in \emph{Proc.\ of the International Conference on Weblogs and
  Social Media, (ISWMC 2012)}, Dublin, Ireland, 2012.

\bibitem{helic2013models}
D.~Helic, M.~Strohmaier, M.~Granitzer, and R.~Scherer, ``Models of human
  navigation in information networks based on decentralized search,'' in
  \emph{Proc.\ of the ACM conference on Hypertext and Social Media}.\hskip 1em
  plus 0.5em minus 0.4em\relax Paris, France: ACM, 2013, pp. 89--98.

\bibitem{horn2013matrix}
R.~Horn and C.~Johnson, \emph{Matrix analysis}, 2nd~ed.\hskip 1em plus 0.5em
  minus 0.4em\relax Cambridge Univ. Press, 2013.

\bibitem{cvetkovic1997eigenspaces}
D.~Cvetkovic, P.~Rowlinson, and S.~Simic, \emph{Eigenspaces of graphs}.\hskip
  1em plus 0.5em minus 0.4em\relax Cambridge University Press, 1997.

\bibitem{strang1993introduction}
G.~Strang, \emph{Introduction to linear algebra}.\hskip 1em plus 0.5em minus
  0.4em\relax Wellesley-Cambridge Press Wellesley, MA, 1993, vol.~3.

\bibitem{menczer2002growing}
F.~Menczer, ``Growing and navigating the small world web by local content,''
  \emph{Proceedings of the National Academy of Sciences}, vol.~99, no.~22, pp.
  14\,014--14\,019, 2002.

\bibitem{lamprecht2016method}
D.~Lamprecht, M.~Strohmaier, and D.~Helic, ``A method for evaluating the
  navigability of recommendation algorithms,'' in \emph{Proc.\ of the
  International Workshop on Complex Networks and their Applications}.\hskip 1em
  plus 0.5em minus 0.4em\relax Springer, 2016, pp. 247--259.

\bibitem{levene2004navigating}
M.~Levene and R.~Wheeldon, ``{Navigating the World Wide Web},'' in \emph{Web
  Dynamics}, M.~Levene and A.~Poulovassilis, Eds.\hskip 1em plus 0.5em minus
  0.4em\relax Springer, 2004, pp. 117--151.

\bibitem{lamprecht2017structure}
D.~Lamprecht, K.~Lerman, D.~Helic, and M.~Strohmaier, ``{How the structure of
  Wikipedia articles influences user navigation},'' \emph{New Review of
  Hypermedia and Multimedia}, vol.~23, no.~1, pp. 29--50, 2017.

\bibitem{kleinberg2000navigation}
J.~Kleinberg, ``Navigation in a small world,'' \emph{Nature}, vol. 406, no.
  6798, p. 845, 2000.

\bibitem{kleinberg2002small}
------, ``Small-world phenomena and the dynamics of information,'' in
  \emph{Proc.\ of the International Conference on Advances in Neural
  Information Processing Systems (NIPS 2001)}, Vancouver, British Columbia,
  Canada, 2002, pp. 431--438.

\bibitem{watts2002identity}
D.~Watts, P.~Dodds, and M.~Newman, ``Identity and search in social networks,''
  \emph{Science}, vol. 296, no. 5571, pp. 1302--1305, 2002.

\bibitem{adamic2001search}
L.~Adamic, R.~Lukose, A.~Puniyani, and B.~Huberman, ``Search in power-law
  networks,'' \emph{Physical Review E}, vol.~64, no.~4, p. 046135, 2001.

\bibitem{huberman1998strong}
B.~Huberman, P.~Pirolli, J.~Pitkow, and R.~Lukose, ``{Strong regularities in
  World Wide Web surfing},'' \emph{Science}, vol. 280, no. 5360, pp. 95--97,
  1998.

\bibitem{west2012human}
R.~West and J.~Leskovec, ``Human wayfinding in information networks,'' in
  \emph{Proc.\ of the International Conference on World Wide Web (WWW
  2012)}.\hskip 1em plus 0.5em minus 0.4em\relax Lyon, France: ACM, 2012, pp.
  619--628.

\bibitem{lamprecht2015improving}
D.~Lamprecht, F.~Geigl, T.~Karas, S.~Walk, D.~Helic, and M.~Strohmaier,
  ``{Improving recommender system navigability through diversification: A case
  study of IMDB},'' in \emph{Proc.\ of the International Conference on
  Knowledge Technologies and Data-driven Business, {I-KNOW} '15}, Graz,
  Austria, 2015, pp. 21:1--21:8.

\bibitem{leicht2006vertex}
E.~Leicht, P.~Holme, and M.~Newman, ``Vertex similarity in networks,''
  \emph{Physical Review E}, vol.~73, no.~2, p. 026120, 2006.

\bibitem{higham2008functions}
N.~Higham, \emph{Functions of matrices: theory and computation}.\hskip 1em plus
  0.5em minus 0.4em\relax SIAM, 2008, vol. 104.

\bibitem{estrada2012physics}
E.~Estrada, N.~Hatano, and M.~Benzi, ``The physics of communicability in
  complex networks,'' \emph{Physics reports}, vol. 514, no.~3, pp. 89--119,
  2012.

\bibitem{foster2001faster}
K.~Foster, S.~Muth, J.~Potterat, and R.~Rothenberg, ``{A faster Katz status
  score algorithm},'' \emph{Computational \& Mathematical Organization Theory},
  vol.~7, no.~4, pp. 275--285, 2001.

\bibitem{das2004some}
K.~Das and P.~Kumar, ``Some new bounds on the spectral radius of graphs,''
  \emph{Discrete Mathematics}, vol. 281, no. 1-3, pp. 149--161, 2004.

\bibitem{stevanovic2015spectral}
D.~Stevanovic, \emph{Spectral radius of graphs}.\hskip 1em plus 0.5em minus
  0.4em\relax Academic Press, 2015.

\bibitem{heath2002scientific}
M.~Heath, \emph{{Scientific Computing}}.\hskip 1em plus 0.5em minus 0.4em\relax
  McGraw-Hill New York, 2002.

\bibitem{han2017closed}
G.~Han and H.~Sethu, ``Closed walk sampler: An efficient method for estimating
  the spectral radius of large graphs,'' in \emph{Proc.\ of the IEEE
  International Conference on Big Data}.\hskip 1em plus 0.5em minus 0.4em\relax
  Boston, USA: IEEE, 2017, pp. 616--625.

\bibitem{young1981rate}
N.~Young, ``The rate of convergence of a matrix power series,'' \emph{Linear
  Algebra and its Applications}, vol.~35, pp. 261--278, 1981.

\bibitem{knuth1989concrete}
D.~Knuth, R.~Graham, and O.~Patashnik, ``Concrete mathematics,'' \emph{Addison
  Wesley}, 1989.

\bibitem{popolizio2008acceleration}
M.~Popolizio and V.~Simoncini, ``Acceleration techniques for approximating the
  matrix exponential operator,'' \emph{SIAM Journal on Matrix Analysis and
  Applications}, vol.~30, no.~2, pp. 657--683, 2008.

\bibitem{saad1992analysis}
Y.~Saad, ``{Analysis of some Krylov subspace approximations to the matrix
  exponential operator},'' \emph{SIAM Journal on Numerical Analysis}, vol.~29,
  no.~1, pp. 209--228, 1992.

\bibitem{guttel2013rational}
S.~Guttel, ``Rational krylov approximation of matrix functions: Numerical
  methods and optimal pole selection,'' \emph{GAMM-Mitteilungen}, vol.~36,
  no.~1, pp. 8--31, 2013.

\bibitem{kunegis2013konect}
J.~Kunegis, ``{Konect: the Koblenz network collection},'' in \emph{Proc.\ of
  the International Conference on World Wide Web (WWW 2013)}.\hskip 1em plus
  0.5em minus 0.4em\relax Rio de Janeiro, Brazil: ACM, 2013, pp. 1343--1350.

\end{thebibliography}

\end{document}